%% file: main_lmcs.tex
\documentclass{CSML}

\def\dOi{12(1:1)2016}
\lmcsheading%
{\dOi}
{1--38}
{}
{}
{Nov.~30, 2014}
{Jan.~06, 2015}
{}

\ACMCCS{[{\bf Theory of computation}]: Models of computation;
  Semantics and reasoning---Program reasoning---Program
  specifications; {\bf Applied computing}]: Enterprise
  computing---Business process management---Business process modeling}

\pdfoutput=1
\usepackage[dvipsnames]{xcolor}
\usepackage{url}
\usepackage{hyperref}
\usepackage{latexsym} \usepackage{xspace}
\usepackage{prooftree}
\usepackage{xcolor}
\usepackage{mathtools} \RequirePackage[dvipsnames]{xcolor}
\usepackage{stmaryrd}
\input{macros} \usepackage{xifthen}
\usepackage{verbatim}
\usepackage{graphicx}
\usepackage[utf8]{inputenc} \usepackage{newunicodechar} \input symbols.tex

\newunicodechar{…}{\ldots}

\mathcode`|="326A
  
\mathcode`!="0021
\mathcode`?="003F

\DeclareMathOperator{\subject}{\mathsf{subj}}
\newcommand{\standard}{\mathop{\mathsf{std}}}

\renewcommand*{\thefigure}{\Alph{figure}}

\newif\iflong
\longtrue
\newenvironment{proof*}
    {\iflong\proof\else\expandafter\comment\fi}%
    {\iflong\endproof\else\expandafter\endcomment\fi}

\newcommand{\shortcut}[1][6mm]{}

\theoremstyle{plain}
\theoremstyle{plain}
\theoremstyle{plain}
\theoremstyle{plain}
\theoremstyle{plain}
\theoremstyle{definition}

\makeatletter
\DeclareRobustCommand*\cal{\@fontswitch\relax\mathcal}
\makeatother

\begin{document} 
\def\paragraph#1{\bigskip\noindent{\it #1}\ }

\title[Type-checking Liveness for Collaborative Processes]{Type-checking Liveness for Collaborative Processes with Bounded and
Unbounded Recursion\rsuper*}

\author[S.~Debois]{S\o ren Debois\rsuper a}
\author[T.~Hildebrandt]{Thomas Hildebrandt\rsuper b}
\author[T.~Slaats]{Tijs Slaats\rsuper c}
\author[N.~Yoshida]{Nobuko Yoshida\rsuper d}

\address{{\lsuper{a,b,c}}IT University of Copenhagen, Rued Langgaards Vej 7,
  2300 Copenhagen S, Denmark
}
\email{\{debois@itu.dk,hilde@itu.dk,tslaats@itu.dk\}}

\address{{\lsuper c}Exformatics A/S,
  Dag Hammerskjölds Allé 13,
  2100 København Ø, Denmark}

\address{{\lsuper d}Imperial College London,
Department of Computing,  
180 Queen's Gate, 
South Kensington Campus,
SW7 2AZ, 
United Kingdom
}
\email{yoshida@doc.ic.ac.uk} 

\thanks{This work supported in part by the Computational Artifacts
project (VELUX 33295, 2014-2017); by the Danish Agency for Science, Technology
and Innovation; by EPSRC EP/K011715/1, EPSRC EP/K034413/1, and EPSRC
EP/L00058X/1, EU project FP7-612985 UpScale, and EU COST Action IC1201
BETTY.}

\titlecomment{{\lsuper*}Full version of Extended Abstract previously
presented at FORTE '14.}

\begin{abstract}
We present the first session typing system guaranteeing request-response liveness properties for possibly non-terminating communicating processes.
The types augment the branch and select types of the standard binary session types 
with a set of required responses, indicating  that whenever a particular label is selected, a set of other labels, its responses, must eventually also be selected. We prove that these extended types are strictly more expressive than standard session types. We provide a type system for a process calculus similar to a subset of collaborative BPMN processes 
with internal (data-based) and external (event-based) branching, message passing, bounded and unbounded looping. We 
prove that this type system is sound, i.e., it guarantees request-response liveness for dead-lock free processes. We exemplify the use of the calculus and type system on a concrete example of an infinite state system.
\end{abstract}

\maketitle

\section{Introduction}
Session types were originally introduced as typing systems for particular
$\pi$-calculi, modelling the interleaved execution of two-party protocols. A
well-typed process is guaranteed freedom from race-conditions as well as
communication compatibility, usually referred to as session
fidelity~\cite{DBLP:conf/esop/HondaVK98,DBLP:journals/entcs/YoshidaV07,DBLP:journals/iandc/Vasconcelos12}.
Session types have subsequently been studied intensely, with much work on
applications, typically to programming languages,
e.g.,~\cite{DBLP:journals/iandc/Dezani-CiancagliniDMY09,hu:2008,DBLP:conf/icdcit/HondaMBCY11,DBLP:conf/coordination/MostrousV11}.
A~number of generalisations of the theory have been proposed, notably to
multi-party session types~\cite{DBLP:conf/popl/HondaYC08}. Multi-party session
types have a close resemblance to choreographies as found in standards for
business process modelling languages such as BPMN~\cite{BPMN20} and WS-CDL, and
have been argued in theory to be able to provide typed BPMN
processes~\cite{DBLP:conf/esop/DenielouY12}.


Behavioral types usually furnish \emph{safety} guarantees, notably progress and lock-freedom \cite{DBLP:journals/corr/abs-1010-5566,DBLP:conf/concur/BettiniCDLDY08,DBLP:conf/coordination/CoppoDPY13,lock4,DBLP:conf/coordination/VieiraV13}.
In contrast, in this paper we extend binary session types to allow specification of
\emph{liveness}---the property of a process eventually ``doing something good''.
Liveness  properties are usually verified by model-checking techniques~\cite{DBLP:journals/fmsd/DammH01,DBLP:conf/dfg/BrillDKWW04,DBLP:conf/esec/CheungGK98}, requiring a state-space exploration.
In the present paper we show that a fundamental class of liveness properties,
so-called \emph{request-response} properties, can be dealt with by type rules,
that is, without resorting to state-space exploration.
As a consequence, we can deal statically with infinite state systems as exemplified below. Also, liveness properties specified in types can be understood and used as interface specifications and for compositional reasoning.
\begin{figure}[htp]
\centering
\shortcut
\includegraphics[width=\textwidth]{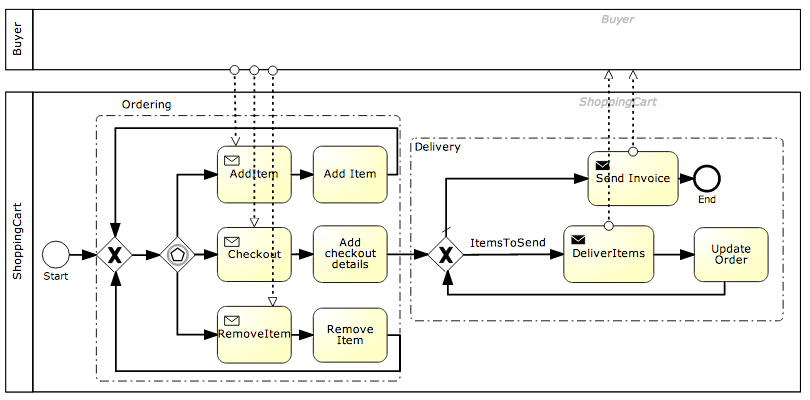} 
\caption{A Potentially Non-live Shopping Cart BPMN Process}   
\shortcut
\label{fig:nonlive}
\end{figure}

As an example, the above  diagram  contains two pools: The Buyer and the
ShoppingCart. Only the latter specifies a process, which has two parts:
Ordering and Delivery. Ordering is a loop starting with an event-based
gateway, branching on the message received by the customer. If it is AddItem
or RemoveItem,  the appropriate item is added or removed from the order,
whereafter the loop repeats. If it is Checkout, the loop is  exited, and the
Delivery phase commences. This phase is again a loop, delivering the ordered
items and then sending the invoice  to the buyer.

A buyer who wants to communicate \emph{safely} with the Shopping Cart, must
follow the protocol described above, and in particular must be able to receive
an unbounded number of items before receiving the invoice. Writing
\newcommand{\AI}{\ensuremath{\mathsf{AI}}}%
\newcommand{\RI}{\ensuremath{\mathsf{RI}}}%
\newcommand{\CO}{\ensuremath{\mathsf{CO}}}%
\newcommand{\DI}{\ensuremath{\mathsf{DI}}}%
\newcommand{\SI}{\ensuremath{\mathsf{SI}}}%
$\AI,\RI,\CO,\DI$, and $\SI$ for the actions ``Add Item'', ''Remove Item'',
``Checkout'', ``Deliver Items'' and ``Send Invoice''; we can describe this
protocol from the point of view of the Shopping Cart with a session type: \[ 
\mu t. \&\{\AI.?.t, \RI.?.t, \CO.?.\mu t'. ⊕\{\DI.!.t',\SI.!.\one\} \}\;.
\]
This session type  can be regarded as a \emph{behavioral} interface, specifying that the process first expects to receive either an $\AI$ (AddItem), $\RI$ (RemoveItem) or a $\CO$ (CheckOut) event. The two first events must be  followed by a message (indicated by ``?''), which in the implementation provides the item to be added or removed, after which the protocol returns to the initial state. The checkout event is followed by a message (again indicated by a ``?'') after which the protocol enters a new loop, either sending a $\DI$ (DeliverItem) event followed by a message (indicated by a ``!'') and repeating, or sending an $\SI$ (SendInvoice) event followed by a message (the invoice) and ending.

However, standard session types can not specify the very relevant
\emph{liveness} property, that a \CO\ (checkout) event is \emph{eventually}
followed by a \SI\ (send invoice) event. This is an example of a so-called \emph{request-response} property: an action (the request) must be followed by a particular response. 
In this paper we conservatively extend binary session types to  specify such request-response properties,  and we show that this extension is strictly more expressive than standard session types. We do so by annotating the checkout selection in the type with the required response:
\[
\mu t. \&\{\AI.?.t, \RI.?.t, \CO{\color{blue}[\{\SI\}]}.?.\mu t'. ⊕\{\DI.!.t',\SI.!.\one\} \}\;.
\]
 Intuitively: ``if \CO\ is selected, then subsequently also \SI\ must be selected.''
\begin{figure}
\centering
\includegraphics[width=.4\textwidth]{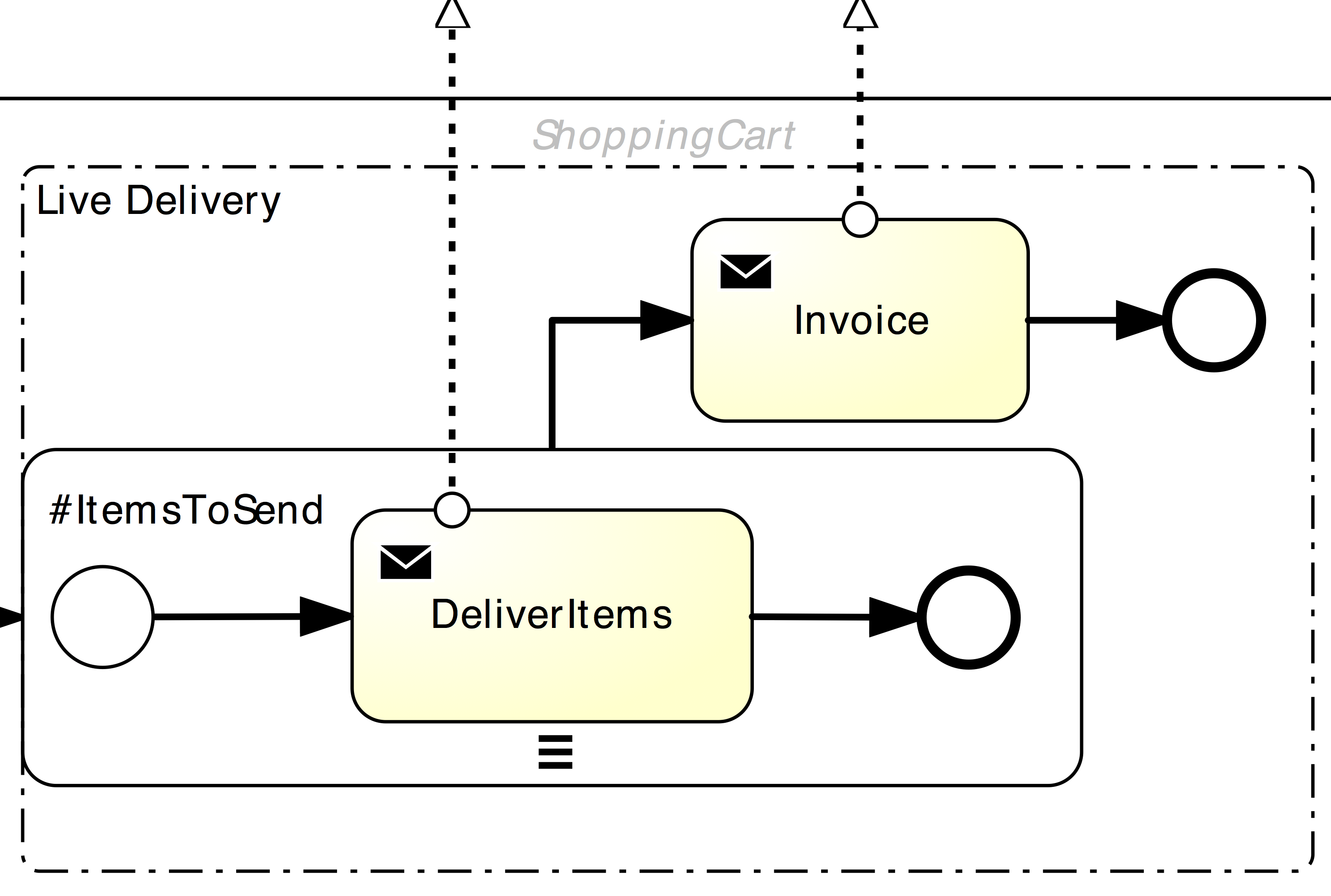} 
\caption{Live delivery with MI Sub-Process}   
\label{fig:live}
\end{figure}

Determining from the flow graph alone if this request-response property is guaranteed is in general not possible: Data values dictate whether the second loop terminates.
However, we can remove this data-dependency by replacing the loop with a %
bounded iteration. In BPMN this can be realised by a Sequential Multiple Instance Sub-process, which sequentially executes a (run-time determined) number of instances of a sub-process. With this, we may re-define  Delivery as in 
Fig.~\ref{fig:live}, yielding a
re-defined Shopping Cart process which has the request-response property. 

In general, we need also be able to check processes where responses are requested within (potentially) infinite loops. 
The type system we present gives such guarantees, essentially by collecting all requested responses in a forward analysis, exploiting that potentially infinite loops can guarantee a particular response only if every path through the loop can; and that order (request-response vs response-request) is in this case irrelevant.
We prove that, if the system is lock free, then the typing system indeed guarantees that all requested responses are eventually fulfilled. Lock-freedom is needed because, as is well known, collaborative processes with interleaved sessions may introduce dependency locks. Lock-freedom is well-studied for both $\pi$-calculus,~e.g., \cite{Kobayashi2002122}, and binary session types~\cite{DBLP:journals/corr/abs-1010-5566,DBLP:conf/concur/BettiniCDLDY08,DBLP:conf/coordination/CoppoDPY13,lock4,DBLP:conf/coordination/VieiraV13}, or may alternatively be achieved by resorting to global types~\cite{DBLP:conf/popl/HondaYC08}.

In summary, our contributions are as follows. 
\begin{itemize}
\item  We extend  binary session types with a notion of \emph{required response}.  
\item  We prove that this extension  induces a strictly more expressive language class than standard session types.
\item  We give a typing system conservatively extending standard binary session types which gives the further guarantee that a lock-free well-typed process will, in any execution, provide all its required responses.
\item We exemplify the use of these extended types to guarantee both safety and liveness properties for a non-trivial, infinite state collaborative process, which exhibits both possibly infinite looping and bounded iteration. 
\end{itemize}

\bigskip\noindent{\it Related work.}
There is a vast amount of work on verification of collaborative processes. Most
of the works take a model-checking approach, where the system under verification is represented as a kind of automaton or Petri Nets. An example that explicitly addresses collaborative business processes is~\cite{DBLP:conf/bpm/RoaCV11}, 
which however does not cover liveness properties. 
Live Sequence Charts (LSCs)~\cite{DBLP:journals/fmsd/DammH01} is a
conservative extension of Message Sequence Charts allowing to distinguish possible
(may) from required (must) behaviour, and thus the specification of liveness
properties. LSCs can be mapped to symbolic timed
automata~\cite{DBLP:conf/dfg/BrillDKWW04} but relies as all model-checking
approaches on abstraction techniques for reducing a large or even infinite
state space to a tractable size. Here the work
in~\cite{DBLP:conf/esec/CheungGK98} is interesting for the fact that the
model-checking can be split on components.  The work
in~\cite{DBLP:conf/lics/KobayashiO09} allows for model-checking of ML programs
by a translation to higher-order recursion schemes. Interestingly, the model-checking problem is reduced to a type-checking problem, but rely on a
technique for generation of a specific type system for the property of
interest. In contrast, our approach is based on a single type system directly
applicable for the process language at hand, where the (less general) liveness
and safety properties of interest are specified as the type to be checked and
can also be used as interface descriptions of processes. The fair subtyping of
\cite{DBLP:conf/icalp/Padovani13}, the only work on session types addressing
liveness we are aware of, details  a  liveness-preserving subtyping-relation
for a session types-like CCS calculus. Here liveness is taken to mean the
ability to always eventually output a special name, whereas in the present
work, we consider the specification of fine-grained request-response liveness
properties---``\emph{if} something happens, something else must happen''.

\paragraph {Overview of this paper.} 
This article presents a full version of an extended abstract that
appeared at FORTE '14 \cite{DBLP:conf/forte/DeboisHSY14}.
The present paper includes the detailed definitions and
explanations, many more examples, and complete proofs. In particular, the
formal development for both basic correctness of the typing system in Section
\ref{section-typing} as well as for the
liveness results in Section~\ref{section-liveness} was mostly absent from the
extended abstract; we believe these, in particular the latter,  to be of
independent interest. 

We proceed as follows. In 
Section~\ref{section-terms} we define  our  
calculus and its LTS-semantics. 
In Section~\ref{section-types} we extend binary session types with specification of request-response liveness properties, give transition semantics for types, 
and sketch a proof that the extended types induce a strictly larger class of languages than does standard types. In Section~\ref{section-typing}
we define exactly how types induce a notion of liveness on processes. In Section~\ref{section-livenesstyping} we give our extended typing rules for sessions with responses and state its subject reduction result. 
In Section~\ref{section-liveness} we 
prove that the extended typing rules guarantees liveness for lock-free processes. 
Finally, in Section~\ref{section-conclusion} we conclude.
\iflong\else
For want of space, this paper omits details and proofs; for these, refer to \cite{full}.
 \fi

 We assume basic familiarity with $\pi$-calculus and binary session types, in
 particular the formulation of the latter in terms of polarised channels. A good
 introductory reference is \cite{Gay-Hole:2005}; an extended discussion of the
 motivations for and ramifications of polarised channels is
 \cite{DBLP:journals/entcs/YoshidaV07}.



\section{Process Terms and Semantics}
\label{section-terms}
Processes communicate only via named communication (session) channels
by synchronizing send and receive actions or synchronizing select and branch events (as in standard session typed $\pi$-calculus). The session typing rules presented
in the next section guarantees that  there is always at most one active send
and receive action for a given channel.
To distinguish dual ends of communication channels, 
we  employ
\emph{polarised names}  
\cite{DBLP:journals/acta/GayH05,DBLP:journals/entcs/YoshidaV07}: If $c $ is
a channel name, $c ^ +$ and $c ^ -$ are the dual ends of the channel
$c$. We call these \emphasise {polarised channel names},  with
``+'' and ``-'' polarities. If $k$ is a polarised  channel name,  we
write $\overline {k}$ for 
\iflong the dual polarised channel name,
\else its dual,
\fi e.g., $\overline {c
^ +} = c ^ -$. 
\iflong
The syntax of
processes is given below.  
\fi
\input{FIGSyntax} 
The first four process constructors are for taking part in a communication. These are standard for  session typed $\pi$-calculi, except that for simplicity of presentation, we only allow data to be sent (see Section~\ref{section-conclusion}). The process $\sendp$ sends data 
 $v$ over channel $k$  when $e\Downarrow v$,
 and proceeds as $P$.  Dually, $\receivep$ receives a data value over channel $k$ and substitutes it for the $x$ binding in $P$. 
 A \emph{branch} process 
 ${\BRANCHS k{l_i}{P_i}}_{i\in I} $ offers a choice between labels $l_i$, proceeding to $P_i$ if the $i$'th
 label is chosen.  
The process \INACT{} is the standard inactive process (termination), and $P | Q$ is the parallel composition of processes $P$ and $Q$.

Recursion comes in two forms: a general, potentially non-terminating
recursion $\recp$, where $X$ binds in $P$; and a primitive
recursion, guaranteed to
terminate, with syntax $\recnp$.  The latter process, when $e$ evaluates to $0$,
evolves to $Q$;  when $e$ evalutes to
$n+1$, it evolves to $P\{n/i\}\{\RECN Xi{e-1}PQ\}$, i.e., the process becomes
$P$ with $n$ substituted for $i$, and the same process except for a decreased  $e$ 
substituted for the process variable $X$. 
\iflong
We assume the following conventions:
\begin{equation}
  \label{conventions}
  \begin{array}{l}
    \text{In $\recnp$, $\INACT$ does not occur in $P$}. \\
    \text{In $\recnp$, no process variable but $X$ occurs free in $P
    $}. \\
    \text{In $\recnp$, there is no sub-term $\REC Y{R}$ or $
    \RECN YieRS$ in $P$}.     \\
    \text{In $\recnp$, there is no sub-term $R|S$ of $P$.}
  \end{array}
\end{equation}
\else
By convention in $\recnp$ neither of  $\INACT$, $\REC Y{P'}$, $\RECN Yi{e'}{P'}{P''}$ and $P'|P''$ occurs
as subterms of $P$. 
\fi
 These conventions 
 ensure that the process $\recnp$ will eventually terminate the loop and execute $Q$. 
Process variables $\VAR X{\tilde k}$ mention the channel names $\tilde k$ active at unfolding time for technical reasons. 
  
 We define the free polarised names $\fn(P)$  of \iflong a process \fi $
 P$ as usual, with $\fn(\varp)=\tilde k$;  substitution of process variables from $\varp\{P/X\} = P$; and finally value substitution $P\{v/x\}$ in the obvious way, e.g., $\sendp\{v/x\} = 
 \SEND k{e\{v/x\}}.(P\{v/x\})$. Variable substitution can never affect channels.

\begin{exa}
\label{example-pi}
We now show how to model the example BPMN process given in the introduction. To illustrate the possibility of type checking infinite state systems, we use a persistent data object 
represented by 
a process $\mathrm{DATA}(o)$ communicating on a session channel $o$. 
\shortcut
\[
\mathrm{DATA}(o)=
\REC X{\;\RECEIVE {o^+}x.\; \REC Y{
	\;o^+?
	\begin{cases} 
		\mathsf{read}.\;\SEND {o^+}x.\;\VAR Y{o^+}\\
		\mathsf{write}.\;\VAR X{o^+}\\
		\mathsf{quit}.\;\INACT
	\end{cases}
	}}
\]
After having received its initial value, this process repeatedly accepts
commands $\mathsf{read}$  and $\mathsf{write}$ on the session channel $o$  for
respectively reading and writing its value, or the command $\mathsf{quit}$ for discarding the data object.

To make examples more readable, we employ the following shorthands. 
We write
$\mathsf{init} (o,v).P$ for $\SEND {o^-}v.P$, which initializes the data object;
we write
$\FREE o.P$ for
$\SEL{o^-}{\mathsf{quit}}P$, the process which terminates the data object session;
we write
 $\READ ox.P$ for $\SEL{o^-}{\mathsf{read}.\; \RECEIVE{o^-}x}P$,
 the process which loads the value of the data object $o$ into the process-local variable $x$;
 and finally, we write $\ASSIGN oe.P$ for $\SEL{o^-}{\mathsf{write}} \SEND{o^-}{e}.P$, the process which
sets the value of the data-object $o$.

The shopping cart process can then be modelled as 
\[
P(Q) = 
\mathrm{DATA}(o) \; \PAR\;
\mathsf{init}(o,\epsilon).\;
\REC X{k}	
	\begin{cases} 
		\AI.\; \RECEIVE kx.\; \READ oy.\; \ASSIGN o{add(y,x)}.\; \VAR X{ko^-}\\
		\RI.\; \RECEIVE kx.\; \READ oy.\; \ASSIGN o{rem(y,x)}.\; \VAR X{ko^-}\\
		\CO.\; \RECEIVE kx.\; \READ oy.\; \ASSIGN o{add(y,x)}.\; Q
	\end{cases}
\]
Here $k$ is the session channel shared with the customer  and $o$ is the session channel for communicating with  the data object modelling order data. We assume our expression language has suitable operators
``add'' and ``rem'' for adding and removing items from the order. Finally, the process $Q$ is a stand-in for either the (non live) delivery part of the BPMN process in Fig.~\ref{fig:nonlive} or the live delivery part shown 
in Fig.~\ref{fig:live}. 
The non-live delivery loop can be 
represented by the process 
\[
D_0 = 
\REC Y \; \READ oy. \; 
 \ifk \, n(y)>0 \;
\begin{array}{l}
\thenk\; 
\SEL {k}{\DI} \; 
\SEND {k}{next(y)}.\; \ASSIGN{o}{update(y)}. \;
\VAR Y{ko^-} \\
\elsek\;
		{\SEL {k}{\SI}\;
		 \SEND {k}{inv(y)}.\;\FREE o. \INACT}
\end{array}
\]
where $n(y)$ is the integer expression computing from the order $y$ the number of items to send, $next(y)$, ${update(y)}$ and ${inv(y)}$  are, respectively, the next item(s) to be sent; an update of the order to mark that these items have indeed been sent; and the invoice for the order.
Whether this process terminates depends on the data operations.  
Using instead bounded iteration, live delivery becomes:
\begin{multline*}
D = \READ oy. \;
\RECN Yi{n(y)}
  {\\\SEL k{\DI} \READ oy. \;
   \SEND k{pickitem(y, i)}. \VAR Y{ko^-}} 
  {\\ \SEL k\SI \;\READ oy.\; \SEND k{inv(y)}.\;
  \FREE o. \INACT}
\end{multline*}
(The second line is the body of the loop; the third line is the continuation.)
Here $pickitem(y,i)$ is the expression extracting the $i$th item from the order $y$.
\qed
\end{exa} 

\bigskip\noindent{\it Transition Semantics.}
\label{section-semantics}
 We give a labelled transition semantics in Fig~\ref{fig:transitions}. We assume a total evaluation relation $e\Downarrow v$;
note 
 the absence  of a  structural congruence. 
  \iflong
 Transition labels 
 for processes are  on one of the following forms. \[%
     λ\grmeq k!v\grmor  k?v\grmor  k⊕
         l\grmor k \& l \grmor \tau \grmor \tau:l
\]
\fi  
We assume $τ$ is  neither a channel nor a polarised channel. Define
     $\subject(k!v) =\subject (k?v) =\subject (k\& l) = \subject(k ⊕ l) = k$
     and $\subject ( τ)= \subject(\tau:l) = τ$, and 
     \iflong 
     define as a technical convenience
     \fi
     $\dual {τ} = τ$.
\def\B{\ref{fig:transitions}}
\begin {figure}
\shortcut[10mm]
\begin{gather*} 
\ltag{\B-Out} %
\frac{e\Downarrow v}{ \sendp\transition {k!v} P }
\;\;\overline k\not\in\fn(P)
\qquad
\ltag{\B-In} %
\frac{}{ \receivep\transition {k?v} P\{v/x\} } %
\;\;\overline k\not\in\fn(P)
\\[4mm] %
\ltag{\B-Sel}
\frac{}{ \selp\transition { k⊕ l} P } 
\;\;\overline k\not\in\fn(P)
\qquad %
\ltag{\B-Bra} %
\frac{}{ \branchps_{i
    \in I}\transition {k\& l_i} P_i } %
    \;\;\overline k\not\in\fn(P_i)
    \quad
\\[4mm] %
\ltag{\B-ParL}    %
\frac{ P
        \transition {λ} Q }{ P |  P'\transition {λ}  Q | P' }
    \qquad\overline{\subject( λ)} \not \in\fn(P')
    \\[4mm] %
\ltag{\B-Com1} %
\frac{ P  \transition {\overline{k}!v}
        P'\qquad Q \transition {k?v} Q' }
    { P | Q\transition {\tau}  P' | Q' } \qquad %
    \ltag{\B-Com2} %
\frac{ P  \transition {\overline{k} ⊕ l} P'\qquad Q
    \transition {k\& l} Q' }{ P | Q\transition {\tau:l}  P' | Q'    }
\\[4mm] %
\ltag{\B-Rec} %
\frac{
    P\{\recp/X\}
    \transition{λ} Q 
}{
    \recp\transition {λ} Q }
\qquad\qquad
\ltag{\B-Prec0} %
\frac{ e\Downarrow 0\qquad Q\transition{ λ}  R }{
\recnp  \transition {λ} R } 
\\[4mm]                    
\ltag{\B-PrecN}  
\frac{ 
   e\Downarrow n+1\qquad P\{n/i\}\{\RECN XinPQ/X\}\transition {λ} R 
}{
    \recnp\transition {λ} R 
} %
\\[4mm]
\ltag{\B-CondT}
\frac{
e \Downarrow \truek\qquad P\transition{λ} P'
}{
\ifp\transition {λ} P'
}
\qquad
\ltag{\B-CondF}
\frac{
e \Downarrow \falsek\qquad Q\transition{λ} Q'
}{
\ifp\transition {λ} Q'
}
\end{gather*}
\caption {Transition semantics for terms}
\label{fig:transitions}
\end{figure}
 We use these rules along with
symmetric rules for \rulename{\B-ParL} and \rulename{\B-Com1/2}.

Compared to standard CCS or $\pi$ semantics, there are two significant changes:
(1) In the \rulename{\B-ParL}, a transition  $λ$ of $P$ is \emph{not} preserved
by parallel 
\iflong composition \fi
if
the co-channel of the subject of $\lambda$ is in $P'$; and (2) in prefix rules,  the co-name of the subject cannot appear in the continuation.
We impose (1) because if
the co-channel of the subject of $\lambda$ is in $P'$, then intuitively
the process $P|P'$ already contains both sides of that session, and so 
does not offer synchronisation on that session to the environment, only 
to $P'$. 
For instance, the process
\[\SEND{c^+}v.Q|\RECEIVE{c^-}x.R\] does \emph{not} have a transition
$\SEND{c^+}v.Q|\RECEIVE{c^-}x.R\transition{c^+!v} Q
|\RECEIVE{c^-}x.R$. Such a transition would be useless: no
environment $U$ able to receive on $c^-$ could be put in parallel with
$P$ and form a well-typed process, because both $U$ and $\RECEIVE{c^-}d.R$ would
contain the name $c^-$ free.
The reason for (2) is similar: If a process $\sendp\transition{k!v
} P$, and $P$ contains $\overline k$,  again no well-typed environment for that process can contain $\overline k$.
For (2), the side-condition could have been expressed as a well-formedness on
syntax in the present setting; however, anticipating a future extension of the
formalism to encompass also delegation, we have chosen the present condition on
semantics instead.

\iflong
 In recent papers
 \cite{DBLP:journals/iandc/Vasconcelos12,DBLP:conf/ppdp/DardhaGS12,DBLP:conf/coordination/VieiraV13},
 session types have been presented  not with polarised names, but rather with seemingly disparate names, connected by a 
 new-name operator, e.g., one writes $( ν x
 y)(\SEL xl | \BRANCHS yl0)$ to form a session with endpoints $x,y$. This latter formulation---while 
 elegant for reduction semantics---is not viable for the present transition
 semantics.  Without the ability to recognise the two ends $c ^ +,  c ^ -$ of a polarised
 channels as either end of a session, we cannot express
the rules \rulename{\B-Par} nor \rulename{\B-Com}.
\fi

\iflong
\begin{lem}
If $P\transition {\lambda} Q$ then $\dual{\subject(\lambda)}\not\in\fn(Q)$. 
\label{lemma-transition-no-coname}
\end{lem}
\begin{proof}
Straightforward induction on the derivation of the transition.
\end{proof}
\fi


\section{Session Types with Responses}
\label{section-types} 

In this section, we 
generalise binary session types to \emph{session types with responses}. In addition to providing the standard communication safety properties, these also allow us to specify request-response liveness properties.

Compared to standard session types, we do not consider delegation (name passing). Firstly, as illustrated by our example calculus, the types are already expressive enough to cover a non-trivial subset of collaborative processes. Secondly, as we show in the end of the section, session types with responses are already strictly more expressive than standard session types with respect to the languages they can express.
Thus, as we also address in Section~\ref{section-conclusion}, admitting delegation and answering the open question about how response obligations can be safely exchanged with the environment, is an interesting direction for future work which is beyond the scope of the present paper. 

We first define request-response liveness
in the abstract. 
  In general, we shall take it to be the property that
``a request is eventually followed by a response''. 
\iflong
For now, we will not
concern ourselves exactly what ``requests'' and ``responses'' are or what

it means for a response to fulfil a request.
\fi
  
\begin{defi} A \emphasise {request-response structure} is a tuple $(A,
    R, \requested,\responded)$ where $ A$ is a set of actions,  $ R$ is a set
    of responses, and $\requested: A\to R$ and $\responded: A\to R$ are maps
    defining the set of responses that an action requests and performs,
    respectively. 
\end{defi}

\paragraph {Notation.} 
\iflong In this setting, response liveness is naturally a property of
sequences. \fi
We write $ε$ for the empty sequence, we let $ φ, ψ$ range over
finite sequences, and we let $α, β,  ɣ$ range over finite or infinite sequences.
We write sequence concatenation by juxtaposition, i.e., $ φα$. 

\begin{defi} Suppose $( A,  R, \requested,\responded)$ is a
    request-response structure and $ α$ a sequence over $ A$.  Then the
     \emph{responses} $\support{ α}$ of $ α$ is defined by $\support{ α}
    =\cup\{\responded(a) | \exists \varphi,\beta.\; α = \varphi a\beta\}$. Moreover, $α$ is \emphasise{live} iff $α = φ  a β  \implies
    \requested(  a ) ⊆  \support { β}$.  \end{defi}

\iflong
\paragraph{Notation.} We shall be specially interested in liveness  of
sequences of transitions. A \emph{finite transition sequence of length $n$} is
a pair of sequences $(s_i)_{i<n}$ and $(t_i)_{i<n-1}$
s.t.~$s_i\transition{t_i}s_{i+1}$ for $i<n$.  An \emph{infinite transition
sequence} is a pair of sequences $(s_i)_{i\in\N}$ and $(t_i)_{i\in\N}$
s.t.~$s_i\transition{t_i}s_{i+1}$. A finite or infinite transition sequence of
a state $s$ is a finite or infinite transition sequence with $s_1=s$.  We
write  $(s_i, t_i)_{i\in\N}$ for infinite sequences and $((s_i, t_i)_{i<n},
s_n)$ for  finite sequences, giving the final state explicitly.
Slightly abusing notation, we sometimes write 
$(s_i,t_i)_{i\in I}$ or even just $(s_i,t_i)$ for a finite or infinite transition sequence, saying that 
it is a transition sequence \emph{of} $s_1$ over $I$.  
\fi

 \begin{defi}[LTS  with responses] Let $(S, L,
    \transition{})$ be an LTS. When the set of labels $L$ is the set of actions of a
    request-response structure $\rho$, we say that $(S, L, \transition{}, \rho) $  is an
    \emphasise {LTS with responses}, and that a 
    transition sequence of this LTS is \emphasise {live} when its underlying
    sequence of labels is.  \end{defi}

\noindent Next,  we present the syntax of types.
\input{FIGTypes}
By convention, the $l_i$ in each $\branchst$ resp.~$\selst$ are distinct. 

 A session type is  a (possibly infinite)  tree of
actions permitted for  one partner of a two-party communication.  The
type $\cbranchst$, called \emph{branch}, is the type of \emph{offering} a choice
between  different continuations.  If the partner chooses the
label $l_i$,  the session proceeds as $T_i$. Compared to standard session types, making the choice $l_i$
 also requests  a subsequent response on every label mentioned in the set of labels $L_i$; we formalise this in the notion of  \emph{responsive trace} below.
Dual to branch is \emph{select} $\cselst$: the type of \emph{making} a choice
between different continuations.  Like branch, making a choice $l_i$ requests every label in
$L_i$ as  future responses.  
The type 
 $\sendvt$ and 
  $\recvvt$ are the types of sending and receiving
data values. 
As mentioned above, channels cannot be communicated. Also, we have deliberately omitted types of values (e.g. integers, strings, booleans) being sent, since this can be trivially added and we want to focus on the behavioural aspects of the types.  
Finally, session types with responses include recursive types. We take
the
equi-recursive  view, identifying a type $T$ and its unfolding into
a potentially  infinite tree. 
%
We define the central notion of \emph{duality} between types
%
as the symmetric relation induced coinductively by
the following rules. \begin {gather} \frac{}{\one\bowtie\one} \quad \frac{
    T\bowtie T'  }{ \sendvt \bowtie \;\recvvt' } \quad \frac{ T_i\bowtie
    T'_i\quad J ⊆ I } { \branchst \bowtie ⊕\{l_j [L'_j]. T'_j \}_{ j\in J} } 
    \label{eq-duality}
\end {gather} 
The first rule says that dual processes agree on when communication ends; the second that if a process sends a message, its dual must receive; and the third says that if one process offers a branch, its dual must choose among the offered choices. However, required responses do not need to match: the two participants in a session need not agree on the notion of
liveness for the collaborative session.

\begin{exa}
\label{ex:tp}
Recall from Ex.~\ref{example-pi} the processes $\mathrm{DATA}(o)$ encoding data-object and $P(D)$ encoding the (live) shopping-cart process. The former treats the channel $o$ as
\[
T_D = \mu t. ?.  \mu s.  \&
\{
		\mathsf{read}. !. s, \;
		\mathsf{write}. t, \;
		\mathsf{quit}.\one \;
		\}\;.
      \] 
      The latter treats its channel $k$ to the buyer as 
\[T_P = \mu t. \&\{\AI.?.t, \; \RI.?.t, \; \CO{\color{blue}[\{\mathsf{SI}\}]}.
?.\mu t'. ⊕\{\DI.!.t', \; \SI.!.\one\} \}\;.
\]
To illustrate both responses in unbounded recursion and 
duality of disparate responses, note that the $P(D)$ actually treats its data
object channel~$o^-$ according to the type $T_{E} = \mu t. !.  \mu s.  \oplus
\{
		\mathsf{read}. ?. s, \;
        \mathsf{write\color{blue}{[\{read\}]}}. t, \;
		\mathsf{quit}.\one \;
		\}
$, i.e., every write is eventually followed by a read. However, $T_D\bowtie T_E$: the types $T_E$ and $T_D$ are nonetheless dual.
\qed
\end{exa}



Having defined the syntax of session types with responses, we proceed to give
their semantics.  The meaning of a session type is the  possible
sequences of  communication actions it allows, requiring that pending
responses eventually be done. Formally, we equip session types with a
labeled transition semantics in Fig.~\ref{figure-basic-type-transitions}. 
\def\A{\ref{figure-basic-type-transitions}}%
\begin{figure}
    \shortcut
\begin {gather*} 
\begin{array}{rl}
  \text{Type transition labels: }\quad&
  ρ \grmeq !\grmor ?\grmor \& l [L]\grmor ⊕ l [L] \\
    \text{Type transition label duality: }\quad&
  !\bowtie ?\quad\text{and}\quad \& l[L]\bowtie  ⊕l[L']\\
  \end{array}
  \\
  \strut
\iflong  \\\fi
    \ltag{\A-Out} \frac{}{ \sendvt
          \transition{!} T } \qquad \frac{}{ \recvvt \transition {?} T}
        \rtag{\A-In} \\[2mm]   \ltag{\A-Bra} \frac{ i\in I }{ \branchst
            \transition{\&l_i[L_i]} T_i } \qquad \frac{ i\in I }{ \selst
                \transition {⊕l_i[L_i]} T_i} \rtag{\A-Sel}
\end{gather*}
\shortcut[4mm]
\caption{Transitions of types (1)} 
\shortcut
\label{figure-basic-type-transitions}
\end{figure}

We emphasise that under  the equi-recursive view of session types,
the  transition system of a recursive type $T$  may in general be infinite.

Taking actions $A$ to be the set of labels ranged over by $ρ$, and recalling
that $\cal L$ is our universe of labels for branch/select, we obtain a
request-response structure $(A,\cal P(\cal L), \requested, \responded)$ with
the latter two operators defined as follows.  
\shortcut[2mm]
\[ \begin{array}{ll} 
 \responded(!) = \responded(?)=  ∅ \quad&
  \responded(\& l [L])  =\responded(⊕ l [L]) = \{l\}\\
\requested(!
        ) = \requested(?)   =  \emptyset & \requested(\& l [L]) =\requested(⊕ l [L]) = L
\end{array} \] 
\iflong In the right-hand column, selecting \else
Selecting \fi a
label  $l$ performs the response $l$; pending responses
$L$ associated with that label are conversely requested. 
The LTS of Fig.~\ref{figure-basic-type-transitions} is thus one with responses, and we may speak of its transition sequences being live or not.

\begin {defi} Let $T$ be a type. We define: \begin{enumerate}
\item  The \emph{traces} $\infinitetraces {T}=\{ ( ρ _i)_{i\in I} | (T_i,ρ_i)_{i\in I}$ 
transition sequence of $T$ \}
\item The \emph{responsive traces}
    $\responsivetraces {T} = \{ α ∈\infinitetraces {T} | 
     α$ live $ \}$.  \end{enumerate} 
\end {defi} 
That is, in
    the responsive traces any  request is followed
    by a response.  
\begin {defi} 
    A type $T$ is  a \emph{standard} session type if it requests
        no responses, that is, every occurrence of $L$ in it is has
        $L=\emptyset$. 
\iflong
        Define an operator $\selected(-)$ as follows, lifting it
        pointwise to sequences.%
\[
    \selected(!)=\selected(?) =  ε
    \qquad
    \selected(\& l [L]) =\selected(⊕  l [ L]) = l 
\]
\else
Define $\selected(\rho)=l$ when $\rho=\& l [L]$ or $\rho =⊕  l [ L]$, otherwise $\epsilon$; lift $\selected(-)$ to sequences by union. 
\fi
We then
    define: \begin{enumerate} 
\item The \emph{selection traces}
            $\selectiontraces{T}=\{\selected( α) | α ∈\infinitetraces {T}\}$
\item The \emphasise{responsive selection traces}
            $\responsiveselectiontraces {T} = \{\selected(α) |  α
                ∈\responsivetraces {T}\}$.  
\item The \emphasize {languages of
                standard session types} \\  ${\cal T}  = \{ 
                \selectiontraces {T} |  \text{$T$ is a standard session type}\}$.  
\item
        The \emphasize {languages of responsive session types} 
        \\ ${\cal R} =\{\responsiveselectiontraces
            {T}| \text{$T$ is a session type with responses}\}$.
\end{enumerate} 
\label{definition-standard-session-type}
\end {defi} 

\noindent That is, we compare standard
        session types and session types of responses by considering the
        sequences of branch/select labels they admit. This follows 
        recent work on multi-party session types and
        automata \cite{DBLP:conf/esop/DenielouY12,DBLP:conf/icalp/DenielouY13}.

\iflong
A fine point: because the $\selected(-)$ map is lifted pointwise and maps ``no
selection'' to the empty sequence $ε$, this definition of languages is oblivious
to send and receive.  E.g, if $φ_S, ψ_T$  are the unique traces of the two
types $S=\;!. ⊕l.⊕ l'.\one$ and $T=  ⊕ l.?.⊕ l'.\one$, then $\selected( φ_S)
=\selected( ψ_T) = ll'$. We formalise this 
insight in the following lemma.
\begin{lem} Let $T$ be a standard session type.   There
    exists a session  type $T'$   with no occurrences of send $\sendvt$ or
    receive $\recvvt$ s.t.~$\selectiontraces{T}=\selectiontraces{T'}$.
    \label{lemma-send-receive-doesnt-matter} \end{lem}
\fi

\begin{exa}
The type $T_P$ of Example~\ref{ex:tp} has (amongst others) the two selection traces:
$t=\AI\,\CO\, \DI\, \DI\, \SI$ 
and
$u=\AI\, \CO\, \DI\, \DI\, \DI\, \cdots$.
Of these, only $t$ is responsive; $u$ is not, since it never selects $\SI$ as required by its $\CO$ action. That is, $t,u\in\selectiontraces{T_P}$ and $t\in\responsiveselectiontraces{T_P}$, but $u\not\in\responsiveselectiontraces{T_P}$.
\qed
\end{exa}

\iflong
\begin{lem}[Session types with responses are deterministic] (1). If
    $T\transition { ρ }$ and $T\transition {  ρ '}$ and $\selected( ρ )
    =\selected( ρ ')\not=\epsilon$, then $ ρ  =  ρ '$.  (2). Consider equally
    long finite transition sequences $((T_i, ρ _i)_{i<n}, T_n)$ and $((S_i, ρ
    '_i)_{i<n}, S_n)$. If $T_1=S_1$ and for each $i<n$ $ ρ _i= ρ '_i$,  then
    also $T_i=S_i$ for each $i\leq n$.  \label{345} \end{lem} 
\begin{proof*}
    (1). Immediate from the convention that each label in a branch or selected
    is distinct.  (2).  By induction on $n$. The base case is trivial.
    For $n=k+1$ we have by the induction hypothesis $T_k=S_k$. By convention, each label in a
    branch or select is distinct, so there is at most one $S$ with
    $T_k\transition{ ρ _k} S$. But then $S=T_{k+1}=S_{k+1}$.  \end{proof*}
\fi

\begin {thm} The language of session types with responses $\cal R$ is
    strictly more expressive than that of standard session types $\cal T$; that
    is,  $\cal T \subset \cal R$.  \label {theorem-expressivity} \end {thm}
\iflong
\begin{proof} The non-strict inclusion is immediate by definition; it remains
    to prove it strict. Consider the following session type with responses,
    $T$.  \[ T \;=\;  μ t. ⊕\left\{ \begin{array}{l} a[b].t \\ b[a].t
\end{array} \right.  \] We shall prove that $\responsiveselectiontraces {T}\not
∈\cal T$.  Suppose not; then there exists a session type $ S$ with
$\selectiontraces {S} =  \responsiveselectiontraces {T}$.  Clearly the
responsive  selection traces $\responsiveselectiontraces {T}$ is the set of
infinite strings over the alphabet $\{ a,b\}$ where both  $ a, b$ occur
infinitely often.  It follows that for all $k > 0$, the string $ a^k$ is a
prefix of an infinite string in $\responsiveselectiontraces{T}$.
We have assumed $\selectiontraces {S} = \responsiveselectiontraces {T}$, so
each $a^k$ must also be a prefix of an infinite string in $\infinitetraces{S}$.
By Lemma~\ref{lemma-send-receive-doesnt-matter}, we may assume $S$ has no
occurrences of send or receive, and so for each $k$ there is a transition
sequence $((S^k_i,  ρ ^k_i)_{i<k},S^k_k)$ with $S^k_1=S$ and $\selected(  ρ
^k_i) = a$.  By induction on $k$ using Lemma \ref{345}, we find that $  ρ ^k_i=
ρ ^{ k +1}_i$ and $S^k_i = S^{k+1}_i $ when $i\leq k$. It follows that $S^i_i =
S^{i+1}_i\transition{  ρ _i} S^{i+1}_{i+1}$ when $i+1\leq k$, and so $(S^i_i,
ρ _i)_{i\in\N}$ is an infinite transition sequence with $S^1_1=S$. But then $
(\selected( ρ  _i))_{i ∈\N} = a ^\omega ∈\selectiontraces{S}$ while clearly not
in $\responsiveselectiontraces{T}$, contradicting
$\responsiveselectiontraces{T}=\selectiontraces{S}$.  
\end{proof}
\else
\begin{proof}[sketch] The non-strict inclusion is immediate by definition; it remains
    to prove it strict. For this consider the session type with responses $T=\mu t. ⊕\{a [b].t; b[a].t\}$, which has as responsive traces
  all strings with both infinitely many $a$s and $b$s. 
  We can find every sequence
 $a^n$ as a \emph{prefix} of such a trace.   But, (by regularity) any \emph{standard} session type that has all $a^n$ as finite traces must also have the trace $a^\omega$, which is not a responsive trace of $T$, and thus the responsive traces of $T$ can not be expressed as the traces of a standard session type.
 \end{proof}
\fi
\section{Session Typing} 
\label{section-typing}
\iflong
    Recall the standard type system for session types,
    presented in Fig.~\ref{fig:std-type-system} with the obvious extension for
    primitive recursion. 
    \input{FIGtype-system-std.tex} 
  \else
Recall that the standard typing system \cite{DBLP:conf/esop/HondaVK98,DBLP:journals/entcs/YoshidaV07} for session types
has judgements $Θ  ⊢_{\standard} P ▹ Δ$.  
 We use this typing system without restating it; 
refer to either \cite{DBLP:conf/esop/HondaVK98,DBLP:journals/entcs/YoshidaV07} or the full version of this paper \cite{full}.
\fi
In this judgement, $ Θ$ takes process variables to session type environments;
in turn, a \emph{session typing environment}~$Δ$ is \label{notation-delta}
    a finite partial map from 
    \iflong 
    channel names and polarised channel names 
    \else
    channels 
    \fi
    to types. 
\begin{eqnarray*}
  \Theta &::= &\epsilon \mid \Theta, X : \Delta \\
  \Delta &::= &\epsilon \mid k : T
\end{eqnarray*}
    We write $Δ, Δ'$ for the union of $Δ$ and $Δ'$, defined when their domains are
    disjoint. We say $Δ$ is
    \emph{completed} if $ Δ (k)=\one$ when defined;
             it is \emph{balanced} if $k: T,
            \overline {k}: U \in Δ$ implies~$T\bowtie U$.

We generalise transitions of types (Fig.~\ref{figure-basic-type-transitions}) to 
session typing environments in Fig.~\ref{figure-type-environment-transitions}, 
with transitions \iflong ranged over by $δ$ as follows;
recall that $ρ$ is a type transition label. \fi
\iflong\[\else $\fi
 δ ::=  τ  | τ: l, L | k:   ρ
 \iflong\] \else $. \fi
We define $\subject (k:ρ) = k$ and  $\subject (τ:  l, L) =\subject (τ) = τ$.
\def\D{\ref{figure-type-environment-transitions}}
\begin {figure}
    \shortcut[9mm]
\begin {gather*} %
\ltag{\D-Lift}
\frac{T
        \transition{ ρ} T'}{k: T\ \transition{k: ρ}  k: T'}
\qquad\qquad
  \ltag{\D-Par}
\frac{ Δ\transition {δ} Δ' }{ Δ, Δ''\transition {δ} Δ', Δ'' }
\\[4mm]
\ltag{\D-Com1} %
\frac{Δ_1
    \transition{k:!} Δ'_1 \qquad Δ_2  \transition{\overline{k}:?} Δ'_2
}{ Δ_1, Δ_2  \transition{ τ} Δ'_1, Δ'_2 }  
  %
\qquad
\ltag{\D-Com2} %
\frac{Δ_1 \transition{k:⊕l[L]} Δ'_1 \qquad Δ_2
  \transition{\overline{k}:\&l[L']} Δ'_2 }
  { Δ_1, Δ_2  \transition{ τ:l,L\cup L'}
  Δ'_1, Δ'_2 } 
\end{gather*}
\shortcut[4mm]
\caption {Transitions of types (2)}
\shortcut
\label {figure-type-environment-transitions}
\end {figure}
We lift $\selected(-),\requested(-)$, and
$\responded(-)$ to actions $δ$ 
\iflong 
as follows. 
 \[ \begin{array}{lll}
        \selected(\tau)= ε \qquad& 
        \selected(k: ρ) =\selected(ρ) \qquad&
         \selected( τ:l, L) = l \\        
         \requested(\tau)= ∅  & 
         \requested( τ: ρ) =\requested(ρ) &
          \requested( τ: l, L) = L \\
        \responded( τ) = ∅ &
        \responded(k: ρ) =\responded(ρ)  &
        \responded(k: l, L) = \{l\}  
        \end{array} 
\]
\else
in the obvious way, e.g., $\requested(\tau:l,L)=L$. 
\fi
    The  type environment transition is thus an LTS
    with responses, and we may speak of its transition sequences being live.

\def\correspond{\simeq}
\begin {defi}
We define a binary relation on type transition labels $δ$ and transition 
labels $λ$, written  
$δ
\correspond
 λ
$,
as follows.
\iflong
\begin {align*}
τ &\correspond τ\quad &
  k:\&l[L] &\correspond k\&l &
  k:!   &\correspond k !v &
\\
τ:l,L & \correspond τ:l &
  k:⊕l[L]  &\correspond k ⊕l &
k:?  &\correspond k ?x\quad&   
\end {align*}
\else
$τ \correspond τ$, 
  $k:\&l[L] \correspond k\&l$, 
  $k:!   \correspond k !v$,
$τ:l,L  \correspond τ:l$,
  $k:⊕l[L]  \correspond k ⊕l $,
$k:?  \correspond k ?x$.
\fi
\end {defi}
\begin{thm}[Subject reduction]
  If $Γ ⊢_{\standard} P ▹ Δ$ and $P\transition {λ} Q$, then there exists
    $δ\correspond \lambda$ s.t.~$ Δ\transition {δ} Δ'$ and $Γ ⊢_{\standard} Q ▹ Δ'$.
    \label{theorem-standard-subject-transition} 
    \end{thm}
    \iflong
The proof is in Appendix~\ref{appendix-standard-subject-reduction}.
\fi

\begin{defi} The \emphasise {typed transition system} is the transition
    system which has states $Γ ⊢_{\standard} P ▹ Δ$ and transitions $Γ ⊢_{\standard} P ▹
    Δ\transition{λ, δ}Γ ⊢_{\standard} P' ▹ Δ'$ whenever there exist transitions
    $P\transition {λ} P'$ and  $Δ\transition {δ} Δ'$ with $\delta\correspond\lambda$. 
\end{defi}
\begin{exa}
\label{ex:typing}
Wrt. the standard session typing system, \emph{both} of
the processes $P(D_0)$ and $P(D)$ of Example~\ref{example-pi}
are typable wrt. the types we postulated for them in 
Example~\ref{ex:tp}. Specifically, we have $\cdot ⊢_{\standard} P(D_0) ▹ k: T_P, o^+:T_D, o^-:\overline{T_D}$ and similarly for $P(D)$. The judgement means that the process $P(D)$ treats $k$ according to $T_P$ and the (two ends of) the data object according to $T_D$ and its syntactic dual $\overline{T_D}$. The standard session typing system of course does not act on our liveness annotations, and so does not care that $P(D_0)$ is not live.
\end{exa}

\iflong
    For the subsequent development, we will need the following lemmas.
    \begin {lem}
     If $Δ\transition {δ} Δ'$ then $\domain (Δ) =\domain (Δ')$.
     \label {lemma-type-transitions-preserve-domains}
    \end {lem}
    \begin {proof*}
    Straightforward induction on the derivation of the transition.
    \end {proof*}

    \begin{lem}
     If $Δ\transition {δ} Δ'$ then either:
    \begin{enumerate} 
    \item
    $δ = k: ρ$ and $ Δ = Δ'', k: T$ and $Δ' = Δ'', k:  T'$ and $ T\transition {ρ} T'$; or
    \item  $δ =  τ$ or $δ = τ: l, L$ and $ Δ = Δ'', k: T,\overline{k}: S$ and $ Δ'=  Δ'',k: T',\overline{k}: S'$ where $ T\transition {ρ} T'$ and $S\transition {ρ'} S'$ and $ ρ\bowtie ρ'$.
    \end{enumerate}
     \label {lemma-type-transition-form}
    \end{lem}
    \begin {proof*}
    Straightforward induction on the derivation of the transition.
    \end {proof*}

    \begin {lem}
     If $Δ\transition {δ} Δ'$ with $Δ$  balanced and $\overline{\subject(δ)}\not ∈\domain (Δ)$, then also $Δ'$ balanced.
     \label {type-transitions-preserve-balance}
    \end {lem}
    \begin {proof*}
    By induction on the   derivation of the transition.

    \Case{\rulename{\D-Lift}} Trivial.

    \Case{\rulename{\D-Par}}   Suppose $Δ, Δ''$ balanced with $\overline{\subject(δ)}\not ∈\domain (Δ, Δ'')$, and consider $k,\overline {k} ∈\domain (Δ', Δ'')$.  If both are in $\domain (Δ'')$, we are done. If both are in $\domain (Δ')$ then by the Lemma~\ref{lemma-type-transitions-preserve-domains}
     they are also in $\domain (Δ)$. Then, because $Δ, Δ''$ balanced implies $Δ$ balanced, we find by the induction hypothesis that also $Δ'$ balanced, whence $(Δ', Δ'') (k)= Δ' (k) \bowtie  Δ' (\overline{k}) =( Δ', Δ'') (\overline {k})$.  Finally,  suppose wlog  $k ∈\domain (Δ')$ and $\overline{k} ∈\domain (Δ'')$; we shall see that this is not possible. By  Lemma~\ref{lemma-type-transition-form} either  
     \begin{equation}
         Δ= Δ_1, k: T\transition{k:ρ}  Δ_1, k: T' = Δ'\qquad\text{with}~T\transition {ρ} T',
         \label{eq:1}
    \end{equation}
    or
    \begin{equation}
    Δ = Δ_2, k: T,\overline{k}: S\transition{\delta}   Δ_2,k: T',\overline{k}: S'
     = Δ'\qquad\text{with}~T\transition {ρ} T'~\text{and}~S\transition {\overline{ρ}} S'. 
    \label {eq:2}
    \end{equation}
    We consider these two possibilities in turn. It cannot be true that
     \eqref {eq:1} holds, because by the assumption $\overline{\subject(δ)}\not ∈\domain (Δ, Δ'')$ we must have $\overline{k}=\overline{\subject( δ)}\not\in\domain ( Δ'')$, contradicting $\overline{k} ∈ \domain (Δ'')$. If instead \eqref{eq:2} holds, then because $\overline{k} ∈ \domain (Δ'')$ then $Δ,  Δ''$ is not defined, contradicting the existence of the transition $ Δ, Δ'' \transition {δ} Δ', Δ''$.
     
     \Case{\rulename{\D-Com1}}  Suppose $Δ_1, Δ_2\transition {τ} Δ_1', Δ_2'$ with $Δ_1, Δ_2$ balanced.
    Using Lem\-ma~\ref{lemma-type-transition-form}, and  \rulename{\D-Com1} we have transitions $ Δ_1 =Δ_3, k: S\transition {k:!} Δ_3, k: S' = Δ_ 1'$ and $ Δ_2 =Δ_4, \overline {k}: T\transition {\overline{k}:?} Δ_4, \overline{k}: T' = Δ_2'$, with  $ S\transition {!} S' $ and $ T\transition {?} T'$. 
     It follows that our original transition is on the form
    \begin{equation*}
    Δ_1, Δ_2=Δ_3, k: S,Δ_4, \overline{k}: T\transition {τ}
    Δ_3, k: S', Δ_4, \overline{k}: T'  =Δ_1', Δ_2'.
    \end{equation*}
    Because $ Δ_1, Δ_2$ balanced then also $Δ_3,Δ_4$ is, and so $Δ_1', Δ_2' = Δ_3, k: S', Δ_4, \overline{k}: T'$ is balanced if $S'\bowtie T'$. But $ S\transition {!} S' $ implies $S =! .S'$ and $ T\transition {?} T'$ implies $T =?.T'$. But we have $!.S' = S\bowtie  T = ?.T'$ by $\Delta_1,\Delta_2$ balanced, and so by definition $ S'\bowtie  T'$.

    \Case{\rulename{\D-Com2}}  We have $Δ_1, Δ_2\transition {τ,l:L\cup L'} Δ_1', Δ_2'$ with $Δ_1, Δ_2$ balanced.
    By Lemma~\ref{lemma-type-transition-form}, and  \rulename{\D-Com2} we have transitions $ Δ_1 =Δ_3, k: S\transition {k: ⊕ l [L]} Δ_3, k: S' = Δ_1'$ and $ Δ_2 =Δ_4, \overline {k}: T\transition {\overline{k}:\& l [L']} Δ_4, \overline{k}: T' = Δ_2'$ and $ S\transition {⊕ l [L]} S' $ and $ T\transition {\& l [L']} T'$. Then our original transition is on the form
    \[
    Δ_1, Δ_2=Δ_3, k: S,Δ_4, \overline{k}: T\transition {τ,l:L\cup L'}
    Δ_3, k: S', Δ_4, \overline{k}: T'  =Δ_1', Δ_2'.
    \]
    Because $ Δ_1, Δ_2$ balanced then also $Δ_3,Δ_4$ is, and so $Δ_1', Δ_2' = Δ_3, k: S', Δ_4, \overline{k}: T'$ is balanced if $S'\bowtie T'$. But
    $ S\transition {k: ⊕ l [L]} S' $ implies $S =⊕\{l_j [L'_j]. S'_j\}_{ j\in J} $ with $l=l_i$ and $L=L_i$ for some $i\in J$, and $S'=S_i'$. Similarly $ T\transition {\& l [L]} T'$ implies 
    $T= \{l_i [L'_i]. T'_i\}_{ i\in I}$ with  $j ∈ I$, $l=l_j$, $L'=L_j$, and $T'=T_j$. Because $ S\bowtie  T$ we may assume $ J ⊆  I$ and $i=j$, whence by definition $ T_i'\bowtie  S_j'$.
    \end {proof*}
\fi
\shortcut

\section{Typing System for Liveness} 
\label{section-livenesstyping}
\iflong In this section, we introduce a variant of the
standard session-typing  system, give intuition for it, and establish its
basic properties, notably subject reduction. In the next section, we shall
prove that this typing system does indeed guarantee  liveness of well-typed
processes.

\else
 We  now give our extended typing system for session types with responses.
\fi
The central judgement will be $  Γ; L ⊢ P ▹
Δ$, with the 
intended meaning that ``with process variables
$Γ$ and pending responses $L$, the process $P$ conforms to
$Δ$.''  We shall see in the
next section that  a well-typed lock-free $P$ is live
 and will eventually perform every response 
in $L$.
\iflong

In detail, here are the environments used in the typing system, along with auxiliary  operations on them.
\else
We need:
\fi
\begin{enumerate} \item  \emphasise {Session typing environments}  $Δ$  defined at the start of Section~\ref{notation-delta}.
 \item\emphasise {Response environments}  $L $ are simply sets of
 branch/select labels. 
   \item \emphasise {Process variable environments} $Γ$ are finite partial maps
     from  process variables $X$ to tuples $(L , L , Δ)$  or $(L, Δ)$.
     We write these  $(  A, I, Δ)$ for (A)ccu\-mu\-la\-ted selections and request
     (I)nvariant.
\iflong
We write $Γ + L$ for the environment satisfying that \[ (Γ +
L)(X)=\left\{\begin{array}{ll}(A  \cup L, I,  Δ)&\mbox{whenever~} Γ (X)=(A
,I,   Δ)\\ Γ (X)&\text{otherwise} \end{array} \right.  \]
We sometimes write $\Gamma + l$ instead of $\Gamma + \{ l \}$.

\else
We define  $(Γ + L)(X) = (A  \cup L, I,  Δ)$ when $ Γ (X) = (A, I, Δ)$ and $ Γ (X)$ otherwise, writing $ Γ + l$ instead of $ Γ +\{l\}$. 
\fi
\end{enumerate}
Intuitively, the \emph{response environment} $L$ is the set of obligations the process
being typed will fulfill; every label in $L$ will eventually be selected by the
process. The \emph{process variable environment} assignes to each process
variable two or three pieces of information. For a general recursion variable $X$, it assigns
three: First an \emph{Accumulator} $A$, which intuitively records the set of
labels selected on a particular path through the ``body'' the recursion on $X$.
Second, an \emph{Invariant} $I$, which intuitively is the set of
labels that will be selected on \emph{any} path through the ``body'' of the
recursion. Finally, $\Delta$ is the usual session environment, associating polarised
channels with session types. For a primitive recursion variable $Y$, we
associate simply the maximal set of responses that is allowed to be pending at
the point of recursion, that is, when the process evolves to simply $Y$.

\input{FIGtype-system}
Our typing system is given in Fig.~\ref{fig:type-system}; we proceed to explain  
the rules. First note that for
finite processes, liveness is ensured if the process terminates with no pending
responses. 
Hence in the rule $\INACT$, the request environment is
required to be empty.
For infinite processes there is no point at which we can insist on having no pending responses. 
\iflong Indeed, 
a
process can be live, meeting its requirements, even though it always has
some pending responses. 
\fi 
\iflong
Take for instance this process, typeable with the type
used in the proof of  Theorem
\ref{theorem-expressivity}. 
\[ 
\else
Consider $
\fi
     \RECtriv X{ \;\SEL ka\;\SEL kb\;\VAR Xk{}} 
\iflong
     \qquad  ▹\qquad
\else 
    $, typeable under $
\fi
     k:μ t.⊕
\iflong
     \left\{ \begin{array}{l} a[b]:t \\ b[a]:t             
                     \end{array} \right.
\;. \]
\else
   \{ a[b].t; b[a].t \} $.
\fi
This process has the single transition sequence 
\iflong \[\else $\fi
    P
\transition {k⊕ a} \SEL kb  P \transition {k ⊕ b} P \transition {k⊕ a}
\cdots
 \iflong \]\else $.\ \fi
At each state but the initial one either $b$ or $a$ is pending.
Yet the process is live:
\emphasise {any response requested in the body of the recursion is also
discharged in the body}, although not in order.  
Since infinite behaviour arises as of unfolding of recursion, 
 responses are ensured if the body of every recursion discharges the requests of that body, even if out of order.
 
 For general recursion, \defr and \varr, we thus find for each recursion a set
 of responses, such that at most that set is requested in its body and that it
 responds with at least that set.  
 In the process variable environment $Γ$ we record 
 this response invariant for
 each variable, along with a tally of the responses performed since the start of the recursion.

 This tally is
 updated by the rules \selr/\branchr\ for select and branch; that is, 
for select, to type  $\selp$ wrt.~$k: ⊕l[L'].T$,
  the process $P$ must do every response in $L'$. In the
 hypothesis, when typing $P$, we add to this the new pending responses $L'$.
 But selecting $l$ performs the response $l$, so
 altogether,  to support pending responses $L$ in the conclusion,
 we must have pending
 responses  $ L ∖\{l\} ∪ L'$  in the hypothesis.  Branching is similar.
  
 The rule for process variable \varr
 typing then 
  checks that the tally includes the invariant, and that the
 invariant includes every currently pending response. 

 For primitive recursion, \defnr and \varnr, we cannot know that the body $P$
 will run at all---$e$ might evaluate to $0$. Thus, it is not helpful to tally
 performed responses; the type system simply checks that at the point of
 recursion \varnr, the pending responses $L$ is contained in the set of
 responses $L'$ that the continuation guarantees to discharge \defnr.
 Notice that this set of $L'$ is allowed to be weakened by the rule \defnr; this is a
 technical necessity for weakening of $L$ to hold in general
 (Lemma~\ref{lemma-L-weakening}).

 Note that the syntactic restriction on the form of the body $P$ in primitive
 recursion \eqref{conventions} ensures that there is no issues with branching: because $P$ can
 contain \emph{only} the process variable $X$, no branching, internal or
 external, can prevent the execution of the body the specified number of times. 

Finally, note the use of the response environments $L$ in the rules \parr and
\concr. For the former, because of our eventual assumption of lock-freedom, a
process $P|Q$ will discharge the set $L$ if $P$ discharges part of $L$ and $Q$
the rest. For the latter, a conditional promising to discharge a set $L$ must
obviously do so whether it takes the ``true'' or ``false'' branch.

 \iflong
 We conclude our walk-through of the rules by showing the typing derivation of
 the above example process.
 \begin{exa}
   Take $\Delta = k : \mu t. \{ a[b].t; b[a].t \}$. Read the derivation below
   bottom-up; easy to miss changes---but not all changes!---are called out
   \grey{with a grey box}.
$$
   \begin{prooftree}
     \[
       \[
         \[
           \{a\}  ⊆ \{a\} ⊆ \{a,b\}\qquad\qquad\dom(\Delta) = k
           \justifies
           X : (\{\grey{a,b}\}, \{a\}, \Delta) \;\; ;\; \{\grey{a}\} 
              \quad\vdash\quad
            \VAR Xk{}
              \quad  ▹\quad
            k:  \mu t. \{ a[b].t;\; b[a].t \} 
           \using {\varr} 
         \]
         \justifies
         X : (\{\grey{a}\}, \{a\}, \Delta) \;\; ;\; \{\grey{b}\} 
            \quad\vdash\quad
          \SEL kb\;\VAR Xk{}
            \quad  ▹\quad
          k:  \mu t. \{ a[b].t;\; b[a].t \} 
         \using {\selr} 
       \]
     \justifies
     X : (\emptyset, \{a\}, \Delta) \;\; ;\; \{a\} 
        \quad\vdash\quad
      \SEL ka\;\SEL kb\;\VAR Xk{}
        \quad  ▹\quad
      k:  \mu t. \{ a[b].t;\; b[a].t \} 
     \using {\selr} 
    \]
       \justifies
      \cdot\;\;; \;\emptyset 
        \quad\vdash\quad
      \RECtriv X{ \;\SEL ka\;\SEL kb\;\VAR Xk{}} 
        \quad  ▹\quad
      k:  \mu t. \{ a[b].t;\; b[a].t \} 
     \using {\defr} 
    \end{prooftree}
$$
\end{exa}

 \fi

\begin{defi}
 We define the \emphasise {standard}  process variable environment $\standard(Γ) $ associated with a process variable environment $Γ$ 
 as follows.
 \[
 \standard(Γ)(X)=\left\{
    \begin{array}{ll}
        Δ\quad &\mbox{whenever~} Γ (X)=(A,I, Δ)\\ 
        Δ&\mbox{whenever~} Γ (X)=(I,   Δ)
        \end{array}
        \right.
 \]
\end{defi}

\begin{thm} If $  Γ; L ⊢ P ▹ Δ$ then also $
    \standard(Γ)  ⊢_{\standard} P  ▹ Δ$.  
\end{thm}
\begin{proof*} 
Straightforward induction on the typing derivation, using for \branchr and \selr\ that $ \standard(Γ + L)=\standard(
Γ)$; for \varr/\defr\ that $\standard( Γ,X: (A,I,  Δ)) = \standard (Γ), X: Δ$; and for \varnr/\defnr\ that $\standard( Γ,X: (I,  Δ))$ $ =  \standard (Γ), X: Δ$.
\end{proof*}

\iflong
    We proceed to establish  basic properties of our typing system, eventually
    arriving at subject reduction. 
First, the typing system allows weakening of the pending responses~$L$.

    \begin {lem} If $ Γ ; L ⊢ P ▹   Δ$ and $L' ⊆  L$, then also $ Γ ; L' ⊢
        P ▹   Δ$.  \label{2} 
        \label{lemma-L-weakening}
    \end {lem} 
    \begin {proof*} 
    By induction on the
    derivation of the typing of $P$.

    \Case{\inactr} We have $ Γ ; L ⊢ \INACT ▹ Δ$.   By typing $L =\emptyset$ and
    our desired property  is  vacously true.

    \Case{\sendr} Immediate from the induction hypothesis.

    \Case{\recvr}   Immediate from the induction hypothesis.

    \Case{\branchr}  We have $ Γ;L  ⊢ \branchps k  ▹  Δ, k:\branchst$.  By typing
    we must have for all $i\in I$ that   $ Γ + l_i;(L  \setminus l_i) \cup L_i
    ⊢ P_i  ▹  Δ , k: T_i	        $.  By the induction hypothesis $ Γ + l_i;(L ' \setminus
    l_i) \cup L_i  ⊢ P_i  ▹  Δ , k: T_i$,  and we conclude $  Γ;L ' ⊢\ \branchps  ▹
    Δ, k:\branchst$.

    \Case{\selr} Similar to \branchr.

    \Case{\concr}  We have $ Γ ; L ⊢ P_1 | P_2 ▹ Δ$.  By typing we have
    $ Γ ; L_i ⊢ P_i ▹ Δ_i$ with $ L = L_1\cup L_2$ and $Δ =
    Δ_1 , Δ_2$. Consider a subset $L' ⊆ L_1 \cup  L_2$. By the induction hypothesis $
    Γ ; L_i\cap  L' ⊢ P_i ▹ Δ_i$ and, noting that $(L_1\cap L')\cup( L_2\cap L') =
    ( L_1\cup  L_2)\cap L' = L'$, we find $ Γ ; L' ⊢ P_1 | P_2 ▹ Δ$

    %
    %
    \Case{\varnr}  Immediate from the premise $  L ⊆   L'$.

    \Case{\defnr}  Immediate from the premise $  L ⊆   L'$.

    \Case{\varr}  Immediate from the premise $  L ⊆   I$.

    \Case{\defr}  Immediate from the  premise $  L ⊆   I$.

    \Case{\ifr} Immediate from the induction hypothesis.
    \end {proof*}


For the subsequent development, we will need to understand how typing
    changes when the environment $\Gamma$ does; in particular what happens when
    $\Gamma(X) = (A,I,\Delta)$ and the accumulator $A$ changes. To this end, we
    define an ordering ``$\leq$'' which captures such environments being the
    same except for the accumulator growing. 

    \begin{defi}
    We define $\Gamma \leq \Gamma'$ iff $\Gamma(X)=(A,I,\Delta)$ implies $\Gamma'(X)=(A',I,\Delta)$ with $A\subseteq A'$ and $\Gamma(X)=(I,\Delta)$ implies $\Gamma'(X)=(I,\Delta)$. 
    \end{defi}

    The key insight about this ordering is then that moving from smaller to
    greater preserves typing, i.e., accumulators admit weakening.

    \begin{lem}
    If $Γ; L ⊢ P ▹ Δ$ and $Γ\leq Γ'$ then also $Γ'; L ⊢ P ▹ Δ$.
    \label{lemma-gamma-monotone-precisely}
    \end{lem}
    \begin{proof*}
    Straightforward induction. We report the two essential cases.

    \Case{\selr}
    We have
    \[
     \frac{
           Γ + l_j; (L  ∖ l_j) ∪ L_j ⊢ P  ▹  Δ, k: T_j
        }{
           Γ; L  ⊢ \SEL{k}{l_j}{P}  ▹  Δ, k:\selst 
        }
    \]
    Noting that $Γ\leq Γ'$ implies $Γ+l_j\leq Γ'+l_j$ we find by IH and \selr
    \[
    \frac{
        Γ' + l_j; (L  ∖ l_j) ∪ L_j ⊢ P  ▹  Δ, k: T_j
        }{
        Γ'; L  ⊢ \SEL{k}{l_j}{P}  ▹  Δ, k:\selst
        }\;.
    \]

    \Case{\defr}
    We have
    \[
     \frac{
         Γ, \AT{X}{(∅, I,   Δ)};  I
        ⊢ P  ▹   Δ \qquad L   ⊆  I
      }{
         Γ; L  ⊢ \recp  ▹   Δ
      }  
    \]
    Noting that $ Γ, X:(∅, I,   Δ)\leq Γ', X:(∅, I,   Δ)$ we have by IH and \defr
    \[
    \frac{
         Γ', \AT{X}{(∅, I,   Δ)};  I
        ⊢ P  ▹   Δ \qquad L   ⊆  I
      }{
         Γ'; L  ⊢ \recp  ▹   Δ
      }  \;.
    \]
    \end{proof*}

    With this insight into the ordering ``$\leq$'', we can now establish that
    extending the environment $\Gamma$ does not change typing.

    \begin{lem}
    If $Γ; L ⊢ P ▹ Δ$ then also $ Γ + L'; L ⊢ P ▹ Δ$.
    \label {lemma-gamma-monotone-typeable}
    \end{lem}
    \begin {proof*}
    Immediate from Lemma~\ref{lemma-gamma-monotone-precisely}. 
    \end {proof*}

    \begin {lem}[Process variable substitution] Suppose that $ Γ  , X:t; L ⊢ P
      ▹   Δ$ where either $t=(A, I, Δ')$ or  $t =(I,Δ')$. Suppose moreover that
      $  Γ  ; I ⊢ Q ▹   Δ'$  where $X$ is not free in
        $Q$.  Then also $ Γ  ; L ⊢ P\{Q/X\} ▹
        Δ$ \label{lemma-process-variable-substitution}
    \end {lem} 
    \begin {proof*} 
    By induction on the typing
        derivation.

    \Case {\inactr}  We have $   Γ  , X:t; L ⊢ \INACT ▹ Δ$. By
    typing $L =\emptyset$. Observe that $\INACT\{Q/X\}=\INACT$. 
    Thus, by \inactr,  we have $  Γ  ; L ⊢\INACT\{Q/X\} ▹ Δ$.

    \Case{\sendr} Immediate from the induction hypothesis.

    \Case{\recvr} Immediate from the induction hypothesis.

    \Case {\branchr} By typing, we have
    \[
     \frac{
           \forall i \in I: \quad  (Γ , X:t) + l_i; (L  ∖ l_i) ∪ L_i  ⊢ P_i  ▹  Δ , k: T_i	        
        }{
             Γ , X:t; L  ⊢\ {\BRANCHS k{l_i}{P_i}}_{i\in I}  ▹  Δ, k:\branchst 
        }
    \]
    Suppose first $t=(A,I, Δ')$. Then $(Γ , X: (A,I, Δ')) + l_i = ( Γ + l_i), X:(A ∪{l_i}, I, Δ')$.  But then we may apply the induction hypothesis and \branchr to obtain
    \begin{equation}
    \frac {
    Γ + l_i; (L  ∖ l_i) ∪ L_i  ⊢ P_i\{Q/X\}  ▹  Δ , k: T_i
    }{
    Γ; L  ⊢\ {\BRANCHS k{l_i}{P_i\{Q/X\}}}_{i\in I}  ▹  Δ, k:\branchst 
    }\;\;\;.
    \label{eq:4}
    \end{equation}
    Suppose instead $t=(I, Δ')$. Then $(Γ , X: (I, Δ')) + l_i = ( Γ + l_i), X:(I, Δ')$, and again we may apply the induction hypothesis and \branchr to obtain \eqref{eq:4}.

    \Case{\selr}  Similar to \branchr.

    \Case{\concr}  We have $   Γ  , X:t;L ⊢ P_1 \PAR P_2  ▹
    Δ$. By typing we find some $L _1 \cup L _2=L$ and $ Δ_1,Δ_2= Δ$ such that $ Γ   , X:t; L_i ⊢  P_i ▹ Δ_i$.  By the induction hypothesis we find $Γ  ; L_i ⊢ P_i\{Q/X\} ▹ Δ_i$, which in turn yields $   Γ;L _1
    \cup L _2 ⊢ (P_1 \PAR P_2)\{Q/X\} ▹  Δ_1,Δ_2$.
    %
    %

    \Case{\varnr} Suppose first $X\not= Y$; then by typing we have 
    \[
     \frac{
    L   ⊆  L '  \qquad  \domain (Δ) =\tilde k  }{
            Γ,Y:(L' ,   Δ), X:t;  L 
           ⊢ \VAR Y{\tilde k}  ▹  Δ     }
        \;\;,
    \]
    so by \varnr also
    \[
          Γ,Y:(L' ,   Δ);  L 
           ⊢ \VAR Y{\tilde k}\{Q/X\}  ▹  Δ\;.
    \]
    If on the other hand $ X = Y$ we have by typing
     \[
     \frac{
    L   ⊆  L '  \qquad  \domain (Δ) =\tilde k   }{
            Γ, X:(L',Δ);  L 
           ⊢ \VAR X{\tilde k}  ▹  Δ
         }
        \;;
    \]
    and it must be the case that  $I=L'$ and $ Δ = Δ'$. We have by assumption $Γ; I ⊢ Q ▹ Δ'$, that is $Γ; L' ⊢ Q ▹ Δ$. By Lemma~\ref{lemma-gamma-monotone-typeable} also $Γ; L ⊢ Q ▹ Δ$, that is, $Γ; L ⊢ \varp\{Q/X\} ▹ Δ$.

    \Case{\defnr}   We have $ Γ, X:(A, I,  Δ'); L  ⊢ \RECN YiePR  ▹  Δ $. By
    typing we have $Γ, X:(A, I,  Δ'), \AT{Y}{(L ',   Δ)};  L'  ⊢ P  ▹  Δ$
    and $Γ;  L'  ⊢ R  ▹  Δ$ for some $L' ⊇ L$.   Using $  Γ  ; I ⊢ Q ▹   Δ'$, by the induction hypothesis $Γ,
    \AT{Y}{(L ',   Δ)};  L'   ⊢ P\{Q/X\}  ▹  Δ$ and $Γ;  L'  ⊢R\{Q/X\}  ▹  Δ$ ,
    which in turn yields $ Γ; L  ⊢  \RECN YiePR\{Q/X\} ▹  Δ $.

    \Case {\varr} 
    Suppose first $X\not= Y$; then by typing we have 
    \[
     \frac{
    L   ⊆  L '  \qquad  \domain (Δ) =\tilde k  }{
            Γ,Y:(A', I', Δ), X:t;  L 
           ⊢ \VAR Y{\tilde k}  ▹  Δ     }
        \;\;,
    \]
    so by \varr also
    \[
          Γ,Y:(A',I', Δ);  L 
           ⊢ \VAR Y{\tilde k}\{Q/X\}  ▹  Δ\;.
    \]
    If on the other hand $ X = Y$ we have by typing
     \[
     \frac{
    L   ⊆ I ⊆ A  \qquad  \domain (Δ) =\tilde k   }{
            Γ, X:(I, A, Δ');  L 
           ⊢ \VAR X{\tilde k}  ▹  Δ'
         }
        \;;
    \]
    where necessarily $\Delta'=\Delta$. 
     We have by assumption $Γ; I ⊢ Q ▹ Δ'$. By Lemma~\ref{lemma-gamma-monotone-typeable} also $Γ; L ⊢ Q ▹ Δ'$, that is, $Γ; L ⊢ \varp\{Q/X\} ▹ Δ'$.

    \Case{\defr}  We have $ Γ  , X:(A, I,  Δ'); L ⊢ \REC YP ▹ Δ$. We find by
    typing $ Γ  , X:(A, I,  Δ'), Y:( A',  I',  Δ); L ⊢ P ▹ Δ$ with $ L ⊆
    I'$, hence by the induction hypothesis $ Γ  ,Y:( A',  I',  Δ); L ⊢ P\{Q/X\} ▹ Δ$, and so by
    \defr $ Γ  ; L ⊢  (\REC YP)\{Q/X\} ▹ Δ$. 

    \Case{\ifr} Immediate from the induction hypothesis.
    \end{proof*}


    \begin {lem}[Term variable substitution]
    If $Γ; L ⊢ P ▹ Δ$ then also $Γ; L ⊢ P\{v/x\} ▹ Δ$.
    \label {lemma-variable-substitution}
    \end {lem}
    \begin {proof*}
    Straightforward induction.
    \end {proof*}

        For technical reasons, we need the following two Lemmas on the presence or
    absence of session names in a typable process. 


    \begin{lem}
    If $Γ; L ⊢ P ▹ Δ$, $k\in\dom(\Delta)$, and $Δ (k) \not=\one$ then $ k∈\fn(P)$.
    \label {lemma-delta-exact-outside-end}
    \end{lem}
    \begin {proof*}
    Straightforward induction.
    \end {proof*}

    \begin{lem}
     Suppose $ Γ; L ⊢ P ▹ Δ, k:T$  with $Δ, k:T$ balanced, $T\not =\one$, and $\overline k\not  ∈\fn (P)$. Then $\overline k\not ∈\domain (Δ)$.
    \label {lemma-1}
    \end{lem}
    \begin{proof*}
     Supposed for a contradiction $\dual k∈\domain (Δ)$. Because $Δ,k:!.T$ balanced, $Δ (\overline {k})\not = \one$. By Lemma~\ref{lemma-delta-exact-outside-end} we thus have $\overline k\in\fn(P)$; contradiction.
    \end{proof*}

    We can now formulate the core lemma which will subsequently be used to prove subject reduction. For the formulation of the lemma, we will slightly abuse notation and consider the range of the $\selected (-)$ operator as an empty or singleton set rather than the empty or singleton sequence as which it was originally defined.

    \begin {lem} 
    Suppose  that $ Γ ; L ⊢ P ▹  Δ$ with 
    $P
    \transition{λ}  Q$. Then there exists a type transition $Δ\transition
    {δ} Δ'$ with $ δ\correspond λ$, such that $
    Γ+\selected( δ); (L∖\responded(δ)) ∪\requested (δ) ⊢ Q ▹ Δ'$.
     Moreover, if  $Δ$ balanced, then also $Δ'$ balanced.
    \label{lemma-subject-reduction} 
    \end {lem} 
    \begin{proof*} By induction on the derivation of the transition.

       \Case{\rulename{\B-Out}} We have $\sendp\transition {k!v} P$ with $\overline{k}\not\in P$ and
       $ Γ; L ⊢ \sendp  ▹  Δ, k: !.T$.  By typing $  Γ
       ; L   ⊢ P  ▹  Δ, k:T $. By \rulename{\D-Lift} we have $k:!.T\transition {k:!} k: T$. 
     By \rulename{\D-Par} $ Δ,k:!.T\transition {k:!}
         Δ,k:T$; 
     Observing that  $k:!\correspond k! v$ and 
     $\responded(k:!) = \selected(k:!) = \requested(k:! )= ∅$ we have found the
       requisite type transition.
       
        Now suppose $Δ, k:!.T$ balanced; we must show $Δ, k: T$ balanced. It is sufficient to show $\dual{k}\not∈ \domain (Δ)$. But  this follows from Lemma~\ref{lemma-1}.

     \Case{\rulename{\B-In}}
     We have $\receivep\transition {k?v} P\{v/x\}$ with $\overline{k}\not\in \fn(P)$ and $ Γ; L  ⊢
     \receivep  ▹ Δ, k: ?.T$. 
      By typing $ Γ ; L  ⊢ P  ▹ Δ, k: T $. 
    By  \rulename{\D-Lift} and \rulename{\D-Par},  $Δ,
      k:?.T\transition {k:?} Δ, k:T$. By Lemma
     \ref{lemma-variable-substitution} we have 
     $ Γ ; L ⊢ P\{v/x\}  ▹ Δ, k: T$.
     Observing  that  $\requested(k:?T[L']) = \selected(k:?T[L']) = \responded(k:?T[L']) = ∅$ and that $k:?\correspond k?  v$
      we have found the requisite transition and typing.
       Preservation of balance follows from Lemma~\ref{lemma-1}.
     
     \Case{\rulename{\B-Bra}} We have $\branchps\transition {k\& l_i} P_i$ and $ Γ; L ⊢{\BRANCHS k {l_i}{P_i}}_{i\in I} ▹ Δ, k:\branchst$.  By typing we have $Γ +
     l_i; (L  ∖ \{l_i\}) ∪ L_i ⊢ P_i ▹  Δ, k: T_i$. By \rulename{\D-Lift} and \rulename{\D-Par} we have $Δ,k:\branchst \transition { k:\& l _i[L_i]}
     Δ,  k: T_ i$. Observing that $\requested(k:\& l_i[L_i])=L_i$,
     $\selected(k:\& l_i[L_i])=\responded(k:\& l_i[L_i])=\{l_i\}$ and that $k:⊕ l_i[L_i]\correspond k \& l_i $, we have found
     the requisite type transition.    
         Preservation of balance  follows from Lemma~\ref{lemma-1}.

    \Case{\rulename{\B-Sel}} We have $\selp\transition { k ⊕ l_i} P$ and $
      Γ; L  ⊢ \SEL{k}{l_i}{P}  ▹  Δ, k:\selst $. By typing $ Γ +
     l_i; (L  ∖ \{l_i\}) ∪ L_i ⊢ P ▹  Δ, k: T_i$. By  \rulename{\D-Lift} and \rulename{\D-Par} we have
     \[  Δ, ⊕\{l_i [L_i].\; T_i\}_{i ∈ I}\transition{k:⊕
     l_i[L_i]}  Δ,T_i\;.\]   Observing that $\requested(k:⊕ l_i[L_i])=L_i$,
     $\selected(k:⊕ l_i[L_i])=\responded(k:⊕ l_i[L_i])=\{l_i\}$ and that $k:⊕ l_i[L_i]\correspond k ⊕ l_i $, we have found
     the requisite type transition.   
          Preservation of balance follows from Lemma~\ref{lemma-1}.

     \Case{\rulename{\B-ParL}} We have $P |  P'\transition {λ}  Q | P'$  with  $\dual {\subject(λ)}\not ∈\fn( P')$ and $ Γ; L ⊢ P |  P' ▹  Δ $.  By typing we
     have for some  $L_1 ∪ L_2 = L$  and $Δ_1 ∪ Δ_2$ that $Γ; L
     _1 ⊢ P  ▹  Δ_1$ and $ Γ; L _2 ⊢ P'  ▹  Δ_2$.  
      By the
     induction hypothesis, we have a transition $ Δ_1\transition {δ}
     Δ_1'$ with  $ Γ +\selected(δ); L_1 ∖\responded(δ) ∪\requested (δ) ⊢ Q ▹ Δ_1'$ and $δ\correspond λ$. By
     Lemma~\ref{lemma-gamma-monotone-typeable} we find also $
     Γ  +\selected(δ); L_2 ⊢ P' ▹ Δ_2$. By Lemma~\ref{lemma-type-transitions-preserve-domains}
     $ \dom(Δ_1) =\dom
     (Δ_1')$ so $Δ'_1, Δ_2$ is defined, and hence by \concr\ we have $
      Γ +\selected(δ) ; L_1 ∖\responded(δ) ∪\requested (δ)   ∪ L_2
     ⊢  Q | P' ▹ Δ'_1, Δ_2$.  This is not exactly the form we need, but observing that
     \[
     (L_1 ∪ L_2) ∖\responded(δ) ∪\requested (δ) ⊆ (L_1 ∖\responded(δ)) ∪\requested (δ)   ∪ L_2,
     \]
    we find again by Lemma~\ref{lemma-L-weakening} that
    $ Γ +\selected(δ) ;  (L_1 ∪ L_2)  ∖\responded(δ) ∪\requested (δ) ⊢  Q | P' ▹ Δ'_1, Δ_2$. By \rulename {\D-Par} $Δ_1, Δ_2\transition {δ} Δ'_1, Δ_2$, and we have found the requisite type transition.

    Now suppose $Δ_1, Δ_2$ balanced. By Lemma~\ref{type-transitions-preserve-balance} it is sufficient to prove that $\dual 
     {\subject  (δ)}\not ∈\domain (Δ_1,Δ_2)$.  If $\subject (δ) = τ$ this is trivial, so say $\subject (δ) = k$ and suppose for a contradiction $\dual k \in \domain(\Delta_1,\Delta_2)$. We must have $\delta = k:\rho$ and because $\delta\correspond\lambda$
     we must have $\subject(\lambda)=\subject(\delta)=k$. 
     By Lemma~\ref{lemma-transition-no-coname} $\dual k\not\in\fn(Q|P')$.  
      By Lemma~\ref{lemma-type-transition-form} we have $ Δ_1 = Δ_1'', k: S$ 
      with $ S\not  =\one$. 
    %
    Because $\Delta_1,\Delta_2$ balanced, $(\Delta_1,\Delta_2)(\dual k)\bowtie S$ and so $(\Delta_1,\Delta_2)(\dual k)\not=\one$. 

    Suppose first $\dual k\in\domain(\Delta_1)$. Then $\dual k\in\domain(\Delta_1'')$, so also $\Delta''(\dual k)\not=\one$, and it follows that
      $
      \Delta_1'(\dual k)=\Delta_1''(\dual k)\not=\one
      $. By Lemma~\ref{lemma-delta-exact-outside-end} $\dual k\in\fn(Q)$, contradicting $\dual k\not\in\fn(Q|P')$. 
      
    Suppose instead $\dual k\in\domain(\Delta_2)$. Then immediately by Lemma~\ref{lemma-delta-exact-outside-end}
    $\dual k\in\fn(P')$, contradicting $\dual k\not\in\fn(Q|P')$.
     
    \Case{\rulename {\B-Com1}}
    We have \[
    \frac{ P_1  \transition {\overline{k}!v}
            P_1'\qquad P_2 \transition {k?v} P_2' }
        { P_1 | P_2\transition {\tau}  P_1' | P_2' } 
    \]
    and
     \[  \frac{
       Γ; L _1 ⊢ P_1  ▹  Δ_1 \qquad Γ; L _2 ⊢ P_2  ▹  Δ_2   }
       {
      Γ; L _1 ∪ L _2 ⊢ P_1 \PAR P_2  ▹  Δ_1,Δ_2
      }\] 
     By the induction hypothesis we find $ Δ_i\transition  { δ_i} Δ'_i $ s.t.~ $
      Γ +\selected(δ_ i) ; L_i ∖\responded(δ_i) ∪\requested(δ_i) ⊢ P_i  ▹  Δ'_i$
      with  $δ_1\correspond\overline {k}  !v$ and $δ_2\correspond k? v$. It
      follows that $δ_1 =k: !$ and $δ_ 2 = k:? $
      whence $\responded (δ_ 1) =\requested (δ_2) =\responded (δ_2) = \requested (δ_1) = ∅ $ and $\selected (δ_1) =\selected (δ_2) =  ε$.
      By Lemma~\ref{lemma-type-transitions-preserve-domains}
      $\Delta_1',\Delta_2'$ defined, and so 
       by \rulename{\D-Com1} we have $ Δ_1, Δ_2\transition {τ} Δ_1', Δ_2'$.
       Noting that $ τ\correspond  τ$ and that $Γ +\selected (δ_1) +\selected (δ_2) = Γ$, we have the required type transition. Since $\subject (τ) =  τ$ and so $\dual{\subject (τ)}\not ∈\domain (Δ_1, Δ_2)$,  it follows from Lemma \ref{type-transitions-preserve-balance} that $ Δ_1', Δ_2'$ is balanced when $Δ_1, Δ_2$ is.
       
     \Case{\rulename{\B-Com2}}
      We have \[ \frac{P_1  \transition {\overline{k} ⊕ l} P'_1\qquad  P_2
      \transition {k\& l} 
      P'_2
    }{ P_1 | P_2 \transition {\tau:l}  P_1' |  P_2'} \]
      and
     \[  \frac{
       Γ; L _1 ⊢ P_1  ▹  Δ_1 \qquad  Γ; L _2 ⊢ P_2  ▹  Δ_2  }{
      Γ; L _1 ∪ L _2 ⊢ P_1 \PAR P_2  ▹  Δ_1,Δ_2
      }\] 
      By induction we find $ Δ_i\transition  { δ_i} Δ'_i $ s.t.~ $
      Γ +\selected(δ_ i) ; L_i ∖\responded(δ_i) ∪\requested(δ_i) ⊢ P_i  ▹  Δ'_i$
      with  $δ_1\correspond\overline {k} ⊕ l$ and $δ_2\correspond k\& l$. It
      follows that for some $L_1', L_2'$ we have $δ_1 =\overline {k}: ⊕ l [L'_1]$ and $δ_ 2 = k:\& l [L'_2]$, and so $\requested (δ_ 1) =  L'_1$ and $\requested (δ_2) = L'_2$, 
       $\responded (δ_1) =\responded (δ_2) =\{l\}$, and
       $\selected (δ_1) =\selected (δ_2) = l$.
    By Lemma~\ref{lemma-type-transitions-preserve-domains}
      $\Delta_1',\Delta_2'$ defined, and so we find a transition $
    Δ_1, Δ_2\transition {τ: l, L_1' ∪ L_2'} Δ_1', Δ_2'$ by \rulename {\D-Com2}.
     By
      \concr\ we find
      $
      Γ + l;
      L_1 ∖\{l\} ∪ L'_1  ∪
      L_2 ∖\{l\} ∪ L'_2  ⊢ 
      P'_1 \PAR P'_2  ▹  Δ'_1,Δ'_2
      $.
       Noting that $ τ: l, L_1' ∪ L_2'\correspond  τ:l$ and 
       \begin{eqnarray*}
       L_1 ∖\{l\} ∪ L'_1  ∪
    L_2 ∖\{l\} ∪ L'_2 &=&  (L_1 ∪ L_2) ∖\{l\} ∪ L_1' ∪ L_2' \\
    &=&
    ( L_1 ∪ L_2) ∖\responded (τ: l, L_1' ∪ L_2') ∪
    \requested (τ: l, L_1' ∪ L_2'),
    \end{eqnarray*}
    we have the required type transition.
    Since $\subject (τ: l, L_1' ∪  L_2') = τ$ and so  $\dual{\subject (τ: l, L_1' ∪  L_2')}\not ∈\domain (Δ_1, Δ_2)$,  it follows from Lemma \ref{type-transitions-preserve-balance} that $ Δ_1', Δ_2'$ is balanced when $Δ_1, Δ_2$ is.

    \Case{\rulename{\B-Rec}}
     We have
    \[ \frac{
                    P\{\recp/X\}\transition{λ} Q }{ \recp\transition {λ} Q }
    \]
    and
    \[
      \frac{
         Γ, \AT{X}{(∅, I,   Δ)};  I
        ⊢ P  ▹   Δ \qquad L   ⊆  I
      }{
         Γ; L  ⊢ \recp  ▹   Δ
      }
    \]
    It follows by \defr\ that also $Γ; I  ⊢ \recp  ▹   Δ$ and by Lemma~\ref{lemma-gamma-monotone-typeable} that 
    $     Γ, \AT{X}{(∅, I,   Δ)};  L    ⊢ P  ▹   Δ$.  It then follows by Lemma~\ref{lemma-process-variable-substitution} that 
    \[
         Γ; L ⊢ P\{\recp/X\}▹ Δ.
    \]
    By the induction hypothesis we find a balance-preserving type transition
    $\Delta\transition{\delta}\Delta'$ with $\delta\correspond\lambda$ and $\Gamma+\selected(\delta);L\setminus{\responded(\delta)}\cup\requested(\delta)
    \proves Q\triangleright \Delta'$. 

    \Case{\rulename{\B-Prec0}}
    We have 
    \[
    \frac{ e\Downarrow 0 \qquad Q\transition{ λ}  R }{
    \RECN XiePQ  \transition {λ} R } 
    \]
    and 
    \[
    \frac{
      Γ, X:(L ',   Δ);  L'
        ⊢ P  ▹  Δ 
        \qquad 
         Γ;  L'  
        ⊢ Q  ▹  Δ 
        \qquad
        L ⊆ L'
      }{
         Γ; L  ⊢ \recnp  ▹  Δ
      }
      \;. 
    \]
     By Lemma~\ref{lemma-gamma-monotone-typeable} we have also
     $  Γ;  L  
        ⊢ Q  ▹  Δ $, and so by the induction hypothesis we find the required balance-preserving type transition.
        
    \Case{\rulename{\B-PrecN}}
    We have 
    \[
    \frac{ 
        e\Downarrow n+1\qquad P\{n/i\}\{\RECN XinPQ/X\}\transition {λ} R 
    }{
        \recnp\transition {λ} R 
    } %
    \]
    and again
    \[
    \frac{
      Γ, X:(L ',   Δ);  L'
        ⊢ P  ▹  Δ 
        \qquad 
         Γ;  L'  
        ⊢ Q  ▹  Δ 
        \qquad
        L ⊆ L'
      }{
         Γ; L  ⊢ \recnp  ▹  Δ
      }
      \;. 
    \]
    By \defnr\ it follows that $Γ; L'  ⊢ \RECN Xi{n}PQ  ▹  Δ$. By Lemmas~\ref{lemma-variable-substitution} and \ref{lemma-L-weakening} we have
        $ Γ, X:(L ',   Δ); L  ⊢ P\{n/i\}  ▹ Δ$. Finally, by Lemma~\ref{lemma-process-variable-substitution} we have 
        \[ Γ; L  ⊢ P\{n/i\}\{\RECN XinPQ/X\} ▹ Δ, \]
    and the requisite balance-preserving type transition follows by the induction hypothesis.

    \Case{\rulename{\B-CondT} and \rulename{\B-CondF}}
     We have,
     \[
     \frac{
    e \Downarrow \truek\qquad P\transition{λ} P'
    }{
    \ifp\transition {λ} P'
    }
    \]
    and
    \[
      \frac{
          Γ; L ⊢ P ▹ Δ\qquad
               Γ; L ⊢ Q ▹ Δ
        }{
         Γ; L ⊢\ifp ▹ Δ
        }\;,
    \]
    and the requisite balance-preserving type transition follows from the induction hypothesis. The other case is the same. 
    \end{proof*}
\fi

\begin {thm}[Subject reduction] 
Suppose  that $  · ; L ⊢ P ▹  Δ$ and $P
\transition{λ}  Q$. Then there exists a type transition $Δ\transition
{δ} Δ'$ with $ δ\correspond λ$, such that $
 · ; (L∖\responded(δ)) ∪\requested (δ) ⊢ Q ▹ Δ'$.
Moreover, if $Δ$ balanced then also $Δ'$ balanced.
\label{theorem-subject-reduction} 
\end {thm} 
\begin{proof*}
Immediate from the Lemma~\ref{lemma-subject-reduction}.
\end{proof*}

\begin{exa}
\label{ex:typingliveness}
With the system of Figure \ref{fig:type-system},
the process $P(D)$ is typable wrt{.} the types given  in
Example~\ref{ex:tp}. The process $P(D_0)$ on the other hand is not: We have $\cdot; \emptyset ⊢ P(D) ▹ k: T_P, o^+:T_D, o^-:\overline{T_D}$, but the same does \emph{not} hold for $P(D_0)$. 
We also exemplify a typing judgment with 
non-trivial guaranteed responses.  The process $D$, the order-fulfillment part of $P(D)$, can in fact be typed
\[
\cdot;\{\SI\}⊢ D ▹ k: \mu t'.\mathop\oplus\{\DI.!.t', \SI.!.\one\},\; o^-:\overline{T_D}
\]
Note the left-most $\{\SI\}$, indicating intuitively that this process will eventually select $\SI$ in \emph{any} execution. 
The process $D$ 
has this property essentially because it is implemented by bounded recursion.
\qed
\end{exa}

\def\A{{\mathcal A}}

\section{Liveness}
\label{section-liveness}
\let\emphasise\emph We now prove that if a lock-free process is
well-typed under our liveness
typing system, then that process  is indeed live.

However, we will need a bit of a detour to \emph{define} lock-free processes: 
Lock-freedom is defined in terms of maximal transition sequences, which are in
turn defined using  (weak) fairness. For both lock-freedom and fairness, we 
will in turn need to speak about occurrences of the
communication actions $\SEND ke$, $\RECEIVE kx$, $\SEL kl$, and $\BRANCH k-$. 

The roadmap for this Section thus becomes:
\begin{enumerate}
  \item In Sub-section~\ref{sub:occ}, we identify key properties of occurrences;
    then use these properties to define fairness, maximal transition sequences,
    and lock-freedom; and finally, what it means for a process to be live. 
  \item In Sub-section~\ref{sub:decomp} 
    We establish decomposition properties of maximal transition sequences for
    parallel compositions, especially wrt.~the responses performed by such
    processes.
  \item 
    We build on these properties and the syntactic restrictions on primtive 
    recursion when in Sub-section~\ref{sub:primrec}, 
    we establish technical properties of primitive recursion. 
  \item Finally, can prove in Sub-section \ref{sub:live} first that well-typed open
    lock-free processes are constrained in what responses they might require by the 
    response invariant of its process variables on the one hand, and the 
    responses it will always do in the other. This is enough to prove that 
    well-typed closed lock-free processes are necessarily live. 
\end{enumerate}
    
\subsection{Occurrences and liveness}
\label{sub:occ}
Since we do not employ a structural congruence, defining occurrences and
liveness is
straightforward, and we simply assume we can identify occurrences across
transitions through residuals---a rigorous treatment is in 
\cite{DBLP:journals/corr/abs-0904-2340}; see also
\cite{DBLP:conf/concur/FossatiHY12}. 
Given an occurrence of an action $a$ in $P$, we 
may thus speak of that occurrence being 
\emph{preserved} by a transition $P\transition{\lambda} Q$ if it remains after
the transition, or being \emph{executed} if it is consumed by the transition.  
Similarly, we will say that an occurrence of an action is 
\emph{enabled} if it is executed by some transition, and 
    \emph{top-level} if it is not nested inside another action.

We will rely on the following properties of occurrences and transitions. 

\begin{lem}  
  \label{lem:occ} 
  Occurrences have the following properties.
\begin {enumerate}
\item  If an occurrence is enabled, it is also top-level.
\item If $P\transition{ λ} Q$ preserves a top-level occurrence of an action $a$
  in $P$, then that occurrence is also top-level in $Q$. 
\item If $P\transition{ λ} Q$ then there exists  an occurrence of an action $a$ in $P$ which is executed by that transition.
\end {enumerate}
    \label{lemma-occurrences} \end{lem}
  \noindent Using occurrences, we define fairness following \cite{DBLP:conf/concur/FossatiHY12}.

\begin {defi} An infinite transition sequence $s = (P_i, λ _i)_{i\in\N}$
  \label{def:occ}
    is \emphasise {fair} iff whenever an occurrence of an action  $a$ occurs enabled in $P_n$ then
    some $m≥ n$ has $P_m\transition{ λ _m} P_{m+1}$ executing that occurrence.
\iflong
\label{def:fairness}
\end {defi}
\begin{exa}
Consider the process 
$
P = (\RECtriv X{ \;\SEL ka\;\VAR Xk{}} )
\PAR
(\RECtriv Y{ \;\SEL ha\;\VAR Yh{}} )
$. This process has an infinite execution
\begin{equation*}
  P\transition {k!a} P \transition {k!a}\cdots\;
  \label{eq:badinftrans}
\end{equation*} 
however, this execution is not fair, since the occurrence of $h!a$ is enabled in the initial
process $P$, but never executed. Conversely, the infinite execution 
\begin{equation}
  P\transition {k!a} P \transition {h!a} P \transition {k!a} \cdots
  \label{eq:inftrans}
\end{equation} 
(note the second label $h!a$) is fair, since, indexing the (identical) processes $P$ by $P=P_1, \ldots$, for
any $i$, the enabled occurrences of actions in $P_i$ are $k!a$ and $h!a$, and these are
executed at either $P_{i+1}$ or~$P_{i+2}$. 
\end{exa}

\begin {defi}
\fi 
\label{def:maximal}
A transition sequence $s$ 
is 
    \emphasise{terminated} iff it has~\iflong finite \fi%
    length $n$ and $P_n\not\transition{\ \ }$.  It 
    is \emphasise {maximal} iff it is finite and terminated or infinite
    and fair.
\iflong
\end {defi}
\noindent
The transition sequence \eqref{eq:inftrans} above is maximal. 

We define lock-freedom in the spirit of \cite{Kobayashi2002122};  notice that the present definition strictly generalises fairness.

\begin {defi}
  \label{def:lock-free}
\fi A maximal transition sequence    $(P_i,λ _i)$ is lock-free iff whenever
there is a top-level occurrence of an action $a$
in $P_i$, then there exists some $j ≥ i$ s.t.~$P_j \transition{λ_j} P_{j+1} $ executes that occurrence.  A process  is lock-free iff all its transition sequences are.
\end {defi}

This definition adapts the one of \cite{Kobayashi2002122} to the setting of
binary session types. However,  the present definition is given in terms of
the transition semantics as opposed to the reduction semantics.  Because of
our amended \rulename{S-ParL} rule,  this makes sense: a process has a
non-$\tau$ transition $ λ$ precisely if it holds only one  half of a session,
and thus needs and expects an environment to perform a co-action in order to
proceed. However, since our definition of lock-freedom  imposes its condition
on \emphasize {every}   maximal transition sequence, there is no presumption
of a cooperative environment.

\begin{exa}
  The infinite execution sequence \eqref{eq:inftrans} above is lock-free.
  Conversely, consider this process.
  \[
    Q = (\RECtriv X{ \;\SEL ka\;\VAR Xk{}} )
        \PAR
        \RECEIVE h x. \SEND {\overline j} x
        \PAR
        \RECEIVE j y. \SEND {\overline h} y
  \]
  Because $\rulename{SPar-L}$ does not permit transitions on neither $h$ nor
  $j$, $Q$ has only one infinite transition sequence:
  \[
    Q \transition {k!a} Q \transition {k!a} \cdots\;.
  \]
  This sequence is maximal but not lock-free. 
  Again indexing the processing in the sequence $Q=Q_1,Q_2,\ldots$ we see that
  at each $Q_i$ we have only one enabled occurrence of an action, namely
  $k!a$, which is immediately consumed by a transition. Hence this sequence is
  maximal. On the other hand, at each $Q_i$ we have top-level occurrences of
  both $h?x$ and $j?y$, but neither is ever executed. Hence this sequence is not
  lock-free. 
\end{exa}

Using the notion of maximal transition sequence, we can now say what it means
for a process to be live.
 \begin{defi}[Live process]
 \label{definition-live-process}
    A well-typed process $Θ ⊢_{\standard} P ▹ Δ$ is \emphasise{live} wrt. $ Θ, Δ$ 
    iff for any maximal transition sequence $(P_i, λ_i)_i$
    of $P$ there exists a live typed transition sequence $ ( Δ_i,δ_i)_i$ of $  Δ$
    s.t.~$(( P_i, Δ_i), (λ_i, δ_i))_i$ is a typed transition sequence of $Θ
    ⊢_{\standard} P ▹ Δ$.  
\end{defi}

\subsection{Decomposition of transition sequences}
\label{sub:decomp}
We proceed to establish properties of transition sequences
of a parallel process: Most importantly, they arise as the  merge of transition
sequences of their underlying left and right processes. First, an auxiliary
definition. 


\begin{defi}
    For a process transition label $\lambda$, define $\selected(\lambda)$ by
\iflong
\begin{gather*}
    \selected(k!v) =\selected (k?v) = \selected ( τ) = \emptyset  \\
    \selected (k\& l) = \selected(k ⊕ l) = \selected(\tau:l) = l
\end{gather*}
\else
    $
    \selected(k!v) =\selected (k?v) = \selected ( τ) = \emptyset 
    $ and $ 
    \selected (k\& l) = \selected(k ⊕ l) = \selected(\tau:l) = l
    $.
\fi
Given a trace $\alpha$  we lift $\selected(-)$ pointwise, that is, $\selected(\alpha)=\{\selected(\lambda) | \alpha=\phi \lambda\alpha'\}$. 
\end{defi}
     
\iflong 
Note that $\delta\correspond\lambda$ implies $\selected(\lambda)=\selected(\delta)$. 

    \begin {lem} 
    \mathcode`|="026A 
    For  any transition sequence $s$ of $P \mid
        Q$, there exists transition sequences $p=(P_i,β_ i)_{i\in |p|}$ and $q=(Q_i,
        δ_i)_{i\in|q|}$ and monotone surjective maps $u:|s|\to| p|$ and $v:|s|\to
        |q|$ such that $s=(( P_{u (i)} | Q_{ v (i)}), α_i)_{i\in|s|}$
        and $\selected(β) ∪\selected( δ) =\selected(α)$.
    \label{lemma-parallel-traces} \end{lem}
    \begin {proof*} 
    \mathcode`|="026A 
    \def\cut{\mathop\mathsf{cut}}
     We prove the existence of such functions for finite $s$; the result for
     infinite $s$ follows.  So suppose $s$ is finite and write it $s=(S_i,
     α_i)_{i\in |s|}.$   
     We proceed by induction on the length of $s$. 
    First, a bit of notation: when $α =α_1\ldots α_n$ we define $\cut( α) = α_1\ldots α_{n -1}$. Now, for  $| s | = 1$, the identity functions suffice.  Suppose instead $| s | = n +1$, and consider the last transition $S_n\transition{ α_n} S_{n+1}$. By the induction hypothesis
     we have transition sequences $p,q$ with labels $ β, δ$ and maps $u,v$ such
     that $S_n = P_{u (n)} \mid Q_{ v (n)}$ and $\selected(β) ∪\selected(δ)
     =\selected(\cut(α))$ etc.  Notice that because $ u, v$ are surjective and monotone, $p,q$ must have lengths $u(n)$ and $v(n)$, respectively. 
      We proceed by cases on  the derivation of this last transition. 
      
    \Case{\rulename{\B-ParL}} 
     We  must in this case have
     \[
     \frac{
     P_{u (n)} \transition{ α_n} R
     }{
      P_{u (n)} \mid Q_{ v (n)}\transition{α_n} R \mid Q_{ v (n)} = S_{n+1}
     }\;.
     \]
     Extend $p$ to $p'$ by taking $P_{u(n)+1}=R$ and $ β_{u(n)} =  α_n$; and extend 
     $u,v$ by taking $u(n+1)=u(n)+1$ and $v(n+1)=v(n) $  and we have found the requisite transition sequences and maps.  It is now sufficient to note that 
     \begin{eqnarray*}
     \selected(α) &=& \selected ( α_n) ∪\selected( \cut(α) ) \\
     &=& 
     \selected  ( β_{u(n)}) ∪\selected( β) ∪\selected(δ) \\
     &=& 
     \selected( β') ∪\selected(δ)
     \;.
     \end{eqnarray*}
     
    \Case{\rulename{\B-Com1}}
     We must have in this case
     \[
     \frac{
     P_{u(n)}\transition{\overline k!v} P'
     \qquad
     Q_{v(n)}\transition{k?v} Q'
     }{
     P_{u(n)}\mid Q_{v(n)}\transition{\tau= α_n} P'\mid Q' = S_{n+1}
     }\;.
     \] 
    Extend $p$ to $p'$ by taking $P_{u(n)+1}=P'$ and $ β_{u(n)} =  \overline k!v$; and similarly extend $q$ to $q'$ by taking $Q_{v(n)+1}=Q'$ and $ δ_{v (n) } = k? v$. Extending also $u,v$ by $u(n+1)=u(n)+1$ and $v(n+1)=v(n)+1$  we have found the requisite transition sequences and maps. It is now sufficient to note that
    \begin{eqnarray*}
     \selected(α) &=& \selected ( τ) ∪\selected( \cut(α)) \\
     &=& 
    \selected( β) ∪\selected(δ) \\
     &=& 
    \selected(\overline k!v)  ∪ \selected( β') ∪ \selected(k?v)  ∪\selected(δ')\\
    &=& \selected( β') ∪\selected(δ') 
     \;.
     \end{eqnarray*}
    \qed
    \Case{\rulename{\B-Com2}}
    We must have in this case
    \[
     \frac{
     P_{u(n)}\transition{\overline k ⊕ l} P'
     \qquad
     Q_{v(n)}\transition{k\& l} Q'
     }{
     P_{u(n)}\mid Q_{v(n)}\transition{\tau:l= α_n} P'\mid Q' = S_{n+1}
     }\;.
     \] 
    Extend $p$ to $p'$ by taking $P_{u(n)+1}=P'$ and $ β_{u(n) } =  \overline k ⊕ l$; and similarly extend $q$ to $q'$ by taking $Q_{v(n)+1}=Q'$ and $ δ_{v (n) } = k\& l$. Extending also $u,v$ by $u(n+1)=u(n)+1$ and $v(n+1)=v(n)+1$  we have found the requisite transition sequences and maps. It is now sufficient to note that
    \begin{eqnarray*}
     \selected(α) &=& \selected ( τ:l) ∪\selected( \cut(α)) \\
     &=& 
    \{l\} ∪\selected( β) ∪\selected(δ) \\
     &=& 
    \selected(\overline  k ⊕ l)  ∪ \selected( β) ∪ \selected(k\& l)  ∪\selected(δ)\\
    &=& \selected( β') ∪\selected(δ') 
     \;.
    \end{eqnarray*}
    \end{proof*}

    \begin {lem} \label{lemma-parallel-preserves-support} 
       If $s$ is a maximal lock-free transition sequence of $ P | Q$ with trace $α$, then there exist maximal lock-free transition sequences $p, q$ of $P, Q$ with traces  $β, δ$, respectively, such that  $\selected(α) =\selected(β)  ∪\support {
    δ}$. 
    \end {lem} 
    \begin{proof*} By Lemma \ref{lemma-parallel-traces} we find transition sequences
    $p=(P_i,β_ i)_{i\in|p|}$ and $q=(Q_i, δ_i)_{i\in|q|}$ and maps $ u, v$ such
    that $s$ can be written $s = (s_i,\alpha_i)_{i\in|s|} = (P_{u (i)} | Q_{ v
    (i)}, \;\alpha_i)_{i\in|s|}$ and $\selected(α) =\selected(β)  ∪\selected( δ)$. It remains
    to prove that these  $p,q$ are maximal and lock-free. 
    Suppose for a contradiction that $p$ is not; the case for $ q$ is similar. Then
    either (A) $p$ is maximal but not lock-free, or (B) $p$ is not maximal. 

    We
    consider first (A); $p$ maximal but not lock-free. Then some top-level
    occurrence of an action $a$ sits in each $P_i$ when $i\geq n$  for some $n$.
    But then for $j\geq u^{-1}(n)$ we must have $s_j=(P_{u(n)},Q_{v(n)})$
    contradicting 
    $s$ lock-free. 

    Consider now (B);  $p$ not maximal.
    Then either (1)  $p$ is finite and can be extended by a
    transition $λ$, or (2) $p$ is infinite
    but not fair.  

     Suppose (1)
    that is, $p$ of finite length $n$ and $P_n
     \transition{λ}$. By  Lemma~\ref{lemma-occurrences}(3) $P_n$ must have an
     occurrence of an enabled action
     $a$. By Lemma~\ref{lemma-occurrences}(1) this occurrence is top-level. But for $i≥ u^{-1}(n)$, 
     $s_i=(P_{u(i)}|Q_{v(i)})$ and so there is a top-level occurrence of $a$ in each such $s_i$, contradicting $P|Q$ lock-free.  

    Suppose instead (2), that is, $p$ infinite but not fair. Then there
    exists a $P_n$ and an occurrence of an enabled action $a$ in $P_n$ s.t.~no $
    β_j$ with $ j ≥  n$ executes  that occurrence.   By definition, every $P_j\transition {β_j} P_{j +1}$
    then preserves that occurrence.  By Lemma~\ref{lemma-occurrences}(1) the occurrence in $P_n$ is top-level, and so by Lemma~\ref{lemma-occurrences}(2) it also is in every $P_j$.
    But for $j≥ u^{-1}(n)$, 
     $s_j=(P_{u(j)}|Q_{v(j)})$, and so we have found a top-level occurrence of $a$ in each such $s_j$, contradicting $P|Q$ lock-free.  
      \end{proof*}

    \def\guarantee#1{\mathcal G(#1)}

    \subsection{Primitive recursion.}
    \label{sub:primrec}
We proceed to establish our main result in steps, starting with \emphasise
{simple} processes. These arise as the body of primitive recursion.

    \begin{defi}
     A process $P$ is \emphasise {simple for $X$} iff
     \begin{enumerate}
     \item no process variable but $X$ occurs free in $P
        $, and
     \item $\INACT$  is not a sub-term of $P$, and
     \item  neither $\REC Y{Q}$ nor $
        \RECN YieQR$ is a sub-term of $P$,
     \item   $Q|R$ is not a sub-term of $P$.   
     \end{enumerate}
    \end{defi}

    \noindent
    Observe that by convention, in $\recnp$, $P$ is simple for $X$.

    \begin{lem}
        If $P$ simple for $X$ and  $s = (P_i, λ_i)_i$ is a  maximal lock-free transition sequence of $P\{Q/X\}$,
        then  $Q\transition{  λ_{j-1}} P_j$ for some $j>1$.
    \label{lemma-simple-substitution}
    \end{lem}
    \begin{proof*}
        \def\sub{\{\tilde v/\tilde x\}}
        \def\Sub{\{Q/X\}}
    By induction on $P$. 

    \Case{``$\INACT$''} Impossible: not simple for $X$. 

    \Case{``$\sendp$''}
    Clearly $(P_ {i +1}, λ_{i +1})_i$ is a maximal lock-free transition sequence of $P\sub$.
    By the induction hypothesis $Q\transition{ λ_j}P_j$ for some $j>2$. 

    \Case{``$\receivep$''}
    Clearly $s' = (P_ {i +1}, λ_{i +1})_i$ is a maximal lock-free transition
    sequence of $P\Sub\sub$.  Because $x$ bound, it is fresh for $Q$, so
    $P\sub\Sub=P\Sub\sub$ and $s'$ is a  maximal lock-free transition
    sequence of the latter. But then by the induction hypothesis 
    $Q\transition{ λ_j}P_j$ for some $j >2$.

    \Case{``$\branchps_{j\in J}$''} 
    Like $\sendp$. 

    \Case{``$\selp$''} Like $\sendp$

    \Case{``$P|R$''} Impossible: not simple for $X$.

    \Case{``$\recp$''}  Impossible: not simple for $X$.

    \Case{``$\RECN YiePR$''} Impossible: not simple for $X$.

    \Case{``$\VAR Y{\tilde k}$''}  By $P$ simple for $X$ we must have $X = Y$ whence
    $s$ is a transition sequence of $X\{Q/X\}=Q$; clearly $Q\transition{\lambda_1} P_2$. 

    \Case{``$\IF ePR $''} 
    Like $\branchps_{j\in J}$.
    \end{proof*}

     \def\sub{\{\tilde Q/\tilde X\}}
    \begin{lem}
     If $s = (P_i, λ_i)_i$ is a maximal lock-free transition sequence of
      $\recnp$ then 
      $Q\transition{  λ_{j-1}} P_j$ for some $j>1$.
    \label{lemma-primitive-recursion-supports-continuation}
    \end{lem}
    \begin{proof*}
     By induction on $n$. 
    If $n=0$ then $s$ is a transition sequence of $Q$ iff it is of $\RECN Xi0PQ$,
    so clearly $Q\transition {λ_1} P_2$.
    If instead $n=m+1$ observe that 
    \[
    \RECN Xi{m+1}PQ\transition {λ_1} R\qquad\text{iff}\qquad
    P\{m/i\}\{\RECN XimPQ/X\}\transition {λ_1} R\;.
     \]
    Take $s'$ to be the same as $s$ except $P_1 =P\{m/i\}\{\RECN XimPQ/X\}$.
    Note that $s'$ is maximal and lock-free. By convention $P$ and so $P\{m/i\}$ is simple for $X$.
    Then by Lemma~\ref{lemma-primitive-recursion-supports-continuation}
    for some $j$ we have  $ Q\transition {λ_{j-1}} P_j$.
    \end{proof*}

 \subsection{Liveness}
    \label{sub:live}
    In this sub-section, we finally establish that well-typed, lock-free
    processes are indeed live. To this end, we will prove that the following
    syntactic function $\A(P)$ is in fact an under-approximation 
    of the set of labels process $P$
    must necessarily select when run. 

\begin{defi}
      \label{def:A}
    When $P$ is a process, we define $\A(P)$ inductively as follows.
    \begin{eqnarray*}
        \A(\INACT) &=& \emptyset\\
        \A(\sendp) &=& \A(P)\\
        \A(\receivep) &=&  \A(P) \\
        \A(\branchps_{i\in I}) &=& \bigcap_{i\in I} \left(\{l_i\}\cup \A(P_i)\right) \\
        \A(\selp) &=& \{l\} \cup \A(P) \\
        \A(P|Q) &=& \A(P) \cup \A(Q) \\
        \A(\recp) &=& \A(P) \\
        \A\recnp &=& \A(Q) 
         \\
        \A(\varp) &=& \emptyset \\
        \A(\ifp) &=& \A(P)\cap \A(Q) 
    \end{eqnarray*}
    \end{defi}

        \begin{prop}
     If $s =(P_i, α_i)_i$ is a  maximal lock-free transition sequence of $P\{Q/X\}$ 
    $\A(P)\subseteq \selected(\alpha)$. 
    \label{proposition-approximation-sound} 
     \end{prop}
     \def\shift{\mathop{\mathsf{shift}}}
     \begin{proof*}
     First, notation: if $α$ is a sequence $ α_1 α_2\cdots$ we define
      $\shift  (α) = α_2\cdots$.
    We proceed by induction on $P$. 

    \Case{``$\INACT$''}
    Immediate from $\A(\INACT)=\emptyset$.

    \Case{``$\sendp$''}
    Clearly $(P_ {i +1}, α_{i +1})_i$ is a maximal lock-free transition sequence of $P\sub$. By the induction hypothesis $ \A ( P) ⊆\selected(\shift( α)) =\selected (α)$.

    \Case{``$\receivep$''}
    Clearly, for some $v$, $(P_ {i +1}, α_{i +1})_i$ is a maximal lock-free transition sequence of $P\sub\{v/x\}$. As $x$ is bound $P\sub\{v/x\}=P\{v/x\}\sub$.
    Using the induction hypothesis $ \A(P)=\A ( P\{v/x\}) ⊆\selected (\shift  (α)) =\support (α)$.

    \Case{``$\branchps_{i\in I}$''}
    Like $\sendp$. 

    \Case{``$\selp$''}
    Like $\sendp$.

    \Case{``$P|R$''}
    By Lemma~\ref{lemma-parallel-preserves-support} there exists traces 
    maximal lock-free  transition sequences $p, q$ of $P\sub, R\sub$ with traces $ β, δ$ s.t.~$\selected(β) \cup\selected(δ) =\selected(α)$.
    Using the induction hypothesis we find $
        \A (P)  ∪ \A (R) 
        ⊆\selected(\beta) ∪\selected(\delta) 
         = \selected(\alpha) 
    $.

    \Case{``$\REC YP$''}
    $s$  is lock-free maximal  transition sequence of $\REC YP\sub$. Then taking $s'$ to be the same as $s$ except $P_1 = P\sub\{\REC Y(P\sub)/Y\} $
     we have a maximal lock-free transition sequence of the latter, 
     also with trace $α$. Using the induction hypothesis $\A(\recp)=\A(P)\subseteq\selected(\alpha)$. 

    \Case{``$\RECNnp YiePR$''}
     $s$ is a lock-free maximal transition sequence of $\RECNnp YIePR\sub$. 
    By Lemma~\ref{lemma-primitive-recursion-supports-continuation}
    for some $j>1$ we have $R\transition{\lambda_{j-1}}P_j$, and so 
    \[
        R\transition{\lambda_{j-1}} P_j \transition {\lambda_j} P_{j+1}\cdots 
     \]
     is a lock-free maximal transition sequence of $R$. 
     The induction hypothesis now implies that $\A(\RECNnp YiePR)=\A(R)
    \subseteq \selected (\shift^{j-2}( α))  ⊆\selected (α)$.

    \Case{``$\VAR Y{\tilde k}$''}
    Immediate from $\A(\VAR Y{\tilde k})=\emptyset$. 

    \Case{``$\IF ePR $''}
    Like $\branchps_{j\in J}$.
    \end{proof*}
      
    %
    %
    %
    %
    %
    %

       Now comes the key lemma: the under-approximation $\A(P)$ is in a sense an
    \emph{over}-approximation of the pending-response environment of a
    well-typed process. The Lemma can be read like this:
    what the process is committed to do ($L$), less what it has done so far
    ($M(\Gamma)$), it will do before iterating ($\A(P)$).

    \begin {lem} 
        Suppose that $ Γ; L ⊢ P ▹  Δ$. 
      Define mappings $M ( ( A, I, Δ) ) =  A$ and $M (   (L, Δ) ) = L$, and 
      \[
            M(Γ)=\bigcup_{X\in \domain( Γ)} M (Γ (X))\;.
        \]
        Then $ L\setminus M (Γ) ⊆\A(P)$.  \label {lemma-discharge}
    \end {lem} 
    \begin {proof*} 
    By induction on the derivation of $ Γ; L ⊢ P
        ▹ Δ$. 
      
    \Case{\inactr}  By typing, $L =\emptyset$.

    \Case{\sendr}  
    $\A(\sendp)=\A(P) ⊇ L ∖ M (Γ)$, the latter by typing and the induction hypothesis.

    \Case{\recvr} Ditto.

    \Case{\branchr} 
    By the induction hypothesis for $i\in I$
    \[
     (L\setminus \{l_i\}) ∪ L_i) \setminus (M (Γ + l_i)⊆\A(P_i)\;.
    \]
    Observe
    that $M (Γ + l_i) = M (Γ) ∪ \{l_i\}$. 
     We compute: 
    \begin{eqnarray*} 
    L\setminus M (Γ)
    &=&
    \cap_{i\in I}(L\setminus M (Γ))\\
    &⊆&
      \cap_{i\in I}( \{l_i\} ∪ (L  \setminus  (M (Γ))))\\ 
     &=&
      \cap_{i\in I}( \{l_i\} ∪ (L  \setminus  (M (Γ)  ∪ l_i)))\\ 
       & =&
     \cap_{i\in I}( \{l_i\} ∪ ((L\setminus\{l_i\})  \setminus  (M (Γ+ l_i))))\\ 
    & ⊆&  
     \cap_{i\in I}(\{l_i\}  ∪ ((L\setminus \{l_i\}) ∪ L_i) \setminus (M (Γ + l_i))) \\ 
     & ⊆ &
    ∩_{i\in I} (\{l_i\}\cup \A(P_i)) \\
    &=& 
    \A(\branchps_{i\in I})
    \end{eqnarray*}

    \Case{\selr}  Similar to \branchr.

    \Case{\concr}
    By typing,  we find
    $ Δ_1, Δ_2$ and $L_1, L_2$ s.t.~ $Γ; L_i ⊢ P_i ▹ Δ_i$.
     By the induction hypothesis we then find that $ L_i\setminus M(Γ)
    ⊆\A(P_i)$. We now compute:
    \begin {eqnarray*} 
    L\setminus M (Γ)  
    & = &  (L_1 ∪ L_2)\setminus  M (Γ)\\
         & = &  L_1\setminus M (Γ) ∪ L_2\setminus M (Γ)\\
          & ⊆ &
        \A(P_1) ∪\A(P_2)\\ & = &\A(P_1|P_2)
    \end {eqnarray*}

    \Case{\varnr} 
    We have $Γ, X: (L', Δ); L ⊢ \varp  ▹ Δ$. 
    By typing $L ⊆ L'$; by definition $L' ⊆ M (Γ)$. But then   $ L ∖M (Γ) = ∅$.

    \Case{\defnr} We have $ Γ; L ⊢\RECNnp XiePQ▹ Δ$. By typing we have $Γ ⊢
    Q ▹ Δ$ and by definition $\A (\RECNnp XiePQ) =\A (Q)  ⊇  L ∖ M (Γ)$,
     the latter by the induction hypothesis.

    \Case{\varr}   
    We have $ Γ,X: (A, I, Δ); L ⊢  X ▹ Δ$. By definition,
    we find $  A ⊆ M( Γ)$, so by typing $ L ⊆ I ⊆ A \subseteq M( Γ)$. But then $L\setminus M (Γ) =\emptyset$. 

    \Case{\defr}  We have $ Γ; L ⊢ \recp ▹ Δ$. By typing we must have $   Γ,X:
    (\emptyset, I, Δ); I ⊢ P ▹ Δ$. We compute.
    \begin{align*} 
    L\setminus M ( Γ) 
        & ⊆  I\setminus (M (Γ) ∪\emptyset)\\ 
        & =  I\setminus M (Γ,X:(\emptyset, I, Δ)) \\ 
        &⊆ \A (P) & &\text {by IH}
    \end{align*} 

    \Case{\ifr} 
    By typing and the induction hypothesis we have $L ∖  Γ (M) ⊆\A (P) $ and $L ∖
    Γ (M) ⊆\A (Q)$.  But then also $L ∖  Γ (M)  ⊆\A (P) ∩\A (Q) =\A
     (\ifp)$. 
    \end {proof*}
\fi

And this is enough: We can now prove that well-typed lock-free processes in fact
do select every response mentioned in the ``pending response'' part of their
type environment.

\begin {prop} 
    Suppose $\cdot\;; L ⊢ P ▹  Δ$ with $P$  lock-free, and let $s=(P_i, α_i)_i$
    be a maximal transition sequence of $P$. Then $ L ⊆\selected (α)$.
    \label{proposition-discharge-empty} 
\end {prop} 
\begin {proof*}
     Observe that necessarily $s$ lock-free. We compute:
     \begin{align*}
         L &⊆\A   (P) &&       \text{By Lemma \ref{lemma-discharge}} \\
           &⊆\selected (α) && 
           \text {By Proposition \ref{proposition-approximation-sound}} 
       \end{align*}
\end {proof*}

\begin{exa}
We saw in Example~\ref{ex:typingliveness} that the process $D$ of Example~\ref{example-pi} is typable $\cdot;\{\SI\}⊢ D ▹ \cdots$. By Proposition~\ref{proposition-discharge-empty} above, noting that $D$ is clearly lock-free, every maximal transition sequence of $D$ must eventually select $\SI$. 
\end{exa}

Applying the definition of ``live process'' we now have our desired Theorem:

\begin{thm} Suppose $ ·\; ; L ⊢ P ▹  Δ$ with $P$ lock-free.
    Then $P$ is live for $·\;,\Delta$.  
    \label{thm:live}
\end{thm}
\begin {proof*} 
    Consider a maximal transition sequence $(P_i,  α_i)$ of $P$. 
    By Definition~\ref{definition-live-process} we must find a live type
    transition sequence $(Δ_i, δ_i)$ of $ Δ$  with $((P_i, Δ_i), ( α_i, δ_i))$
    a typed transition sequence of $ ·\; ⊢ P ▹ Δ$.

    By induction and Theorem \ref{theorem-subject-reduction} there exists a
    sequence $ (Δ_i, L_i, δ_i)_i$ with $ ·\; L_i ⊢ P_i ▹ Δ_i$ 
    and $ Δ_i\transition {δ_i} Δ_{i+1}$ and $ δ_i \correspond  α_i $, and moreover
    $L_{i+1}=L_i\setminus\responded(δ_i) \cup \requested (δ_i)$. Suppose $l\in
    \requested( δ_n)$.  Then $l\in L_{n+1}$. Clearly $P_{n+1}$
    also lock-free, so by Proposition~\ref{proposition-approximation-sound},
    $l\in\selected(\shift^{n}( α))$. That means there exists $j>n$ with
    $l\in\selected( α_j)$. But $  α_j\correspond δ_j$ so $l\in\responded  (δ_j)$.
\end{proof*}

\begin{exa}
\label{ex:typinglivenesslive}
We saw in Example \ref{ex:typingliveness} that $P(D)$ is typable as $\cdot; \emptyset ⊢ P(D) ▹ k: T_P, o^+:T_D, o^-:\overline{T_D}$. Noting $P(D)$ lock-free, by the above Theorem it is live, and so will uphold the liveness guarantee in $T_P$: if $\CO$ is selected, then eventually also $\SI$ is selected. Or in the intuition of the example: If the buyer performs ``Checkout'', he is guaranteed to subsequently receive an invoice. 
\end{exa}

\section{Conclusion and Future Work}
\label{section-conclusion}
We introduced a conservative generalization of binary session types to
\emph{session types with responses}, which allows to specify request-response liveness properties.
We showed that session types with responses
 are 
 strictly more expressive (wrt. the classes of behaviours they can express) 
 than standard binary session types. We provided a typing system for a process calculus similar to a non-trivial subset of collaborative BPMN processes with possibly infinite loops and bounded iteration and proved that lock-free, well typed processes are live. 
 

We have identified several interesting directions for future work: 
\begin{itemize}
  \item The present typing system requires guessing
suitable invariants $I$ for typing checking recursion, i.e., $\recp$. We believe
that the syntactic approximation $\A(P)$ of Definition~\ref{def:A} is the unique maximal
$I$ that will allow typing\footnote{
  We extend our gratitude to Anonymous Referee \#1 for pointing us in this direction.
}.
Lifting this belief to a Theorem would be
an essential foundational step necessary for operationalising the present
typing systems into a type-inference algorithm. 

\item 
The present techniques should be augmented by or combined with existing type systems
for ensuring lock-freedom of session-typed processes (e.g.,
\cite{lock1,lock2,lock3,lock4,lock5,lock6}). 

\item The present work could
be lifted to  multi-party session types,  which guarantees lock-freedom. 

\item The notion of request-response structure invites even more general notion
  of liveness, e.g., rather than requiring a particular future response, one
  might require at least one of a set of possible future responses.

\item Channel passing is presently omitted for simplicity of presentation and not needed for our
expressiveness result (Theorem \ref{theorem-expressivity}).
Introducing it, raises the 
question of 
whether one can delegate the responsibility for doing responses or not? If \emph{not}, then channel passing does not affect the liveness 
properties of a lock-free process, 
and so is not really interesting for the present paper. If 
one \emph{could}, it must be ensured that responses are not forever delegated without ever being fulfilled, which is an interesting challenge for future work. We hope to 
leverage existing techniques for the $\pi$-calculus, e.g.,~\cite{Kobayashi2002122}. 

\item Finally, and more speculatively, we plan to investigate relations to fair subtyping~\cite{DBLP:conf/icalp/Padovani13} and Live Sequence Charts~\cite{DBLP:journals/fmsd/DammH01}.
%

\end{itemize}

\section*{Acknowledgement}
We gratefully acknowledge helpful comments from anonymous reviews at LMCS and
various conferences. 

\iflong
\appendix
\section{Subject reduction proof for the standard session typing system}
\label{appendix-standard-subject-reduction} 

\begin {lem}[Process variable substitution] Suppose that $ Θ  , X:Δ' ⊢_{\standard} P ▹   Δ$. Suppose moreover that $  Θ   ⊢_{\standard} Q ▹   Δ'$  with $X$ is not free in
    $Q$.  Then also $ Θ   ⊢_{\standard} P\{Q/X\} ▹
    Δ$ \label{app-lemma-process-variable-substitution}
\end {lem} 
\begin {proof} 
By induction on the typing
    derivation.

\Case {\inactr}  We have $   Θ  , X: Δ' ⊢_{\standard} \INACT ▹ Δ$. By typing $ Δ$ completed, so by \inactr,  we have $  Θ   ⊢_{\standard}\INACT\{Q/X\} ▹ Δ$.

\Case{\sendr} Immediate from the induction hypothesis.

\Case{\recvr} Immediate from the induction hypothesis.

\Case {\branchr} We  have
\[
 \frac{
       \forall i \in I: \quad  Θ , X: Δ' ⊢_{\standard} P_i  ▹  Δ , k: T_i	        
    }{
         Θ , X: Δ' ⊢_{\standard}\ {\BRANCHS k{l_i}{P_i}}_{i\in I}  ▹  Δ, k:\branchst 
    }
\]
 By the induction hypothesis and \branchr we have 
\begin{equation*}
\frac {
Θ  ⊢_{\standard} P_i\{Q/X\}  ▹  Δ , k: T_i
}{
Θ ⊢_{\standard}\ {\BRANCHS k{l_i}{P_i\{Q/X\}}}_{i\in I}  ▹  Δ, k:\branchst 
}\;\;\;.
\end{equation*}

\Case{\selr}  Similar to \branchr.

\Case{\concr}  We have $   Θ  , X: Δ' ⊢_{\standard} P_1 \PAR P_2  ▹
Δ$. By typing we find some  $ Δ_1,Δ_2= Δ$ such that $ Θ   , X: Δ ⊢_{\standard}  P_i ▹ Δ_i$.  By induction we find $Θ   ⊢_{\standard} P_i\{Q/X\} ▹ Δ_i$, which in turn yields $   Θ 
⊢_{\standard} (P_1 \PAR P_2)\{Q/X\} ▹  Δ_1,Δ_2$.

\Case{\varnr} Suppose first $X\not= Y$; then we have 
\[
 \frac{
\domain (Δ) =\tilde k  }{
        Θ,Y: Δ, X: Δ' 
       ⊢_{\standard} \VAR Y{\tilde k}  ▹  Δ     }
    \;\;,
\]
so by \varnr also
\[
      Θ,Y:   Δ
       ⊢_{\standard} \VAR Y{\tilde k}\{Q/X\}  ▹  Δ\;.
\]
If on the other hand $ X = Y$ we have by typing
 \[
 \frac{
\domain (Δ) =\tilde k   }{
        Θ, X:Δ' 
       ⊢_{\standard} \VAR X{\tilde k}  ▹  Δ
     }
    \;;
\]
and it must be the case that $ Δ = Δ'$. We have by assumption $Θ ⊢_{\standard} Q ▹ Δ'$, that is $Θ ⊢_{\standard} \varp\{Q/X\} ▹ Δ$. 

\Case{\defnr}   We have $ Θ, X:  Δ' ⊢_{\standard} \RECN YiePR  ▹  Δ $. By
typing we have $Θ, X: Δ', \AT{Y}{   Δ} ⊢_{\standard} P  ▹  Δ$
and $Θ, X: Δ' ⊢_{\standard} R  ▹  Δ$. Using $  Θ   ⊢_{\standard} Q ▹   Δ'$, by induction $Θ,
\AT{Y}{  Δ} ⊢_{\standard} P\{Q/X\}  ▹  Δ$ and $Θ ⊢_{\standard}R\{Q/X\}  ▹  Δ$ ,
which in turn yields $ Θ ⊢_{\standard}  \RECN YiePR\{Q/X\} ▹  Δ $.

\Case {\varr} 
Suppose first $X\not= Y$; then by typing we have 
\[
 \frac{
 \domain (Δ) =\tilde k  }{
        Θ,Y:Δ, X: Δ' 
       ⊢_{\standard} \VAR Y{\tilde k}  ▹  Δ     }
    \;\;,
\]
so by \varr also
\[
      Θ,Y: Δ 
       ⊢_{\standard} \VAR Y{\tilde k}\{Q/X\}  ▹  Δ\;.
\]
If on the other hand $ X = Y$ we have by typing
 \[
 \frac{
 \domain (Δ) =\tilde k   }{
        Θ, X: Δ'
       ⊢_{\standard} \VAR X{\tilde k}  ▹  Δ
     }
    \;,
\]
where necessarily $Δ = Δ'$, so $Θ ⊢_{\standard} \varp\{Q/X\} ▹ Δ'$.

\Case{\defr}  We have $ Θ  , X: Δ' ⊢_{\standard} \REC YP ▹ Δ$. We find by
typing $ Θ  , X: Δ', Y: Δ ⊢_{\standard} P ▹ Δ$, hence by the induction hypothesis $ Θ  ,Y: Δ' ⊢_{\standard} P\{Q/X\} ▹ Δ$, and so by
\defr $ Θ   ⊢_{\standard}  (\REC YP)\{Q/X\} ▹ Δ$. 

\Case{\ifr} Immediate from the induction hypothesis.
\end{proof}

\begin {lem}
If $Θ ⊢_{\standard} P ▹ Δ$ then also $Θ ⊢_{\standard} P\{v/x\} ▹ Δ$ .
\label {app-lemma-variable-substitution}
\end {lem}
\begin {proof}
Straightforward induction.
\end {proof}

\begin{lem}
If $Θ ⊢_{\standard} P ▹ Δ$ and $Δ (k) \not=\one$ then $ k∈\fn(P)$.
\label {app-lemma-delta-exact-outside-end}
\end{lem}
\begin {proof}
Straightforward induction.
\end {proof}

\begin{lem}
 Suppose $ Θ ⊢_{\standard} P ▹ Δ, k:T$  with $Δ, k:T$ balanced, $T\not =\one$, and $\overline k\not  ∈\fn (P)$. Then $\overline k\not ∈\domain (Δ)$.
\label {app-lemma-1}
\end{lem}
\begin{proof}
 Supposed for a contradiction $\dual k∈\domain (Δ)$. Because $Δ,k:!.T$ balanced, $Δ (\overline {k})\not = \one$. By Lemma~\ref{app-lemma-delta-exact-outside-end} we thus have $\overline k\in\fn(P)$; contradiction.
\end{proof}

\begin {lem} 
Suppose  that $ Θ  ⊢_{\standard} P ▹  Δ$ with 
$P
\transition{λ}  Q$. Then there exists a type transition $Δ\transition
{δ} Δ'$ with $ δ\correspond λ$, such that $
Θ ⊢_{\standard} Q ▹ Δ'$.
 Moreover, if  $Δ$ balanced, then also $Δ'$ balanced.
\label{app-lemma-subject-reduction} 
\end {lem} 
\begin{proof} By induction on the derivation of the transition.

   \Case{\rulename{\B-Out}} We have $\sendp\transition {k!v} P$ with $\overline{k}\not\in P$ and
   $ Θ ⊢_{\standard} \sendp  ▹  Δ, k: !.T$.  By typing $  Θ
    ⊢_{\standard} P  ▹  Δ, k:T $. By \rulename{\D-Lift} we have $k:!.T\transition {k:!} k: T$. 
 By \rulename{\D-Par} $ Δ,k:!.T\transition {k:!}
     Δ,k:T$.
 Observing that  $k:!\correspond k! v$  we have found the
   requisite type transition.
   
    Now suppose $Δ, k:!.T$ balanced; we must show $Δ, k: T$ balanced. It is sufficient to show $\dual{k}\not∈ \domain (Δ)$. But  this follows from Lemma~\ref{app-lemma-1}.

 \Case{\rulename{\B-In}}
 We have $\receivep\transition {k?v} P\{v/x\}$ with $\overline{k}\not\in \fn(P)$ and $ Θ ⊢_{\standard}
 \receivep  ▹ Δ, k: ?.T$. 
  By typing $ Θ  ⊢_{\standard} P  ▹ Δ, k: T $. 
By  \rulename{\D-Lift} and \rulename{\D-Par},  $Δ,
  k:?.T\transition {k:?} Δ, k:T$. By Lemma
 \ref{app-lemma-variable-substitution} we have 
 $ Θ  ⊢_{\standard} P\{v/x\}  ▹ Δ, k: T$.
 Observing   $k:?\correspond k?  v$
  we have found the requisite transition and typing.
   Preservation of balance follows from Lemma~\ref{app-lemma-1}.
 
 \Case{\rulename{\B-Bra}} We have $\branchps\transition {k\& l_i} P_i$ and $ Θ ⊢_{\standard}{\BRANCHS k {l_i}{P_i}}_{i\in I} ▹ Δ, k:\branchst$.  By typing we have $Θ +
 l_i ⊢_{\standard} P_i ▹  Δ, k: T_i$. By \rulename{\D-Lift} and \rulename{\D-Par} we have $Δ,k:\branchst \transition { k:\& l _i[L_i]}
 Δ,  k: T_ i$. Observing that $k:⊕ l_i[L_i]\correspond k \& l_i $, we have found
 the requisite type transition.    
     Preservation of balance  follows from Lemma~\ref{app-lemma-1}.

\Case{\rulename{\B-Sel}} We have $\selp\transition { k ⊕ l_i} P$ and $
  Θ ⊢_{\standard} \SEL{k}{l_i}{P}  ▹  Δ, k:\selst $. By typing $ Θ +
 l_i ⊢_{\standard} P ▹  Δ, k: T_i$. By  \rulename{\D-Lift} and \rulename{\D-Par} we have \[  Δ, ⊕\{l_i [L_i].\; T_i\}_{i ∈ I}\transition{k:⊕
 l_i[L_i]}  Δ,T_i.\;\]   Observing  that $k:⊕ l_i[L_i]\correspond k ⊕ l_i $, we have found
 the requisite type transition.   
      Preservation of balance follows from Lemma~\ref{app-lemma-1}.

 \Case{\rulename{\B-ParL}} We have $P |  P'\transition {λ}  Q | P'$  with  $\dual {\subject(λ)}\not ∈\fn( P')$ and $ Θ ⊢_{\standard} P |  P' ▹  Δ $.  By typing we
 have for some  $L_1 ∪ L_2 = L$  and $Δ_1 ∪ Δ_2$ that $Θ 
 _1 ⊢_{\standard} P  ▹  Δ_1$ and $ Θ ⊢_{\standard} P'  ▹  Δ_2$.  
  By the
 induction hypothesis, we have a transition $ Δ_1\transition {δ}
 Δ_1'$ with  $ Θ  ⊢_{\standard} Q ▹ Δ_1'$ and $δ\correspond λ$.  By Lemma~\ref{lemma-type-transitions-preserve-domains}
 $ \dom(Δ_1) =\dom
 (Δ_1')$ so $Δ'_1, Δ_2$ is defined, and hence by \concr\ we have $
  Θ   
 ⊢_{\standard}  Q | P' ▹ Δ'_1, Δ_2$ and thus the requisite transition.

Now suppose $Δ_1, Δ_2$ balanced. By Lemma~\ref{type-transitions-preserve-balance} it is sufficient to prove that $\dual 
 {\subject  (δ)}\not ∈\domain (Δ_1,Δ_2)$.  If $\subject (δ) = τ$ this is trivial, so say $\subject (δ) = k$ and suppose for a contradiction $\dual k \in \domain(\Delta_1,\Delta_2)$. We must have $\delta = k:\rho$ and because $\delta\correspond\lambda$
 we must have $\subject(\lambda)=\subject(\delta)=k$. 
 By Lemma~\ref{lemma-transition-no-coname} $\dual k\not\in\fn(Q|P')$.  
  By Lemma~\ref{lemma-type-transition-form} we have $ Δ_1 = Δ_1'', k: S$ 
  with $ S\not  =\one$. 
Because $\Delta_1,\Delta_2$ balanced, $(\Delta_1,\Delta_2)(\dual k)\bowtie S$ and so $(\Delta_1,\Delta_2)(\dual k)\not=\one$. 

Suppose first $\dual k\in\domain(\Delta_1)$. Then $\dual k\in\domain(\Delta_1'')$, so also $\Delta''(\dual k)\not=\one$, and it follows that
  $
  \Delta_1'(\dual k)=\Delta_1''(\dual k)\not=\one
  $. By Lemma~\ref{lemma-delta-exact-outside-end} $\dual k\in\fn(Q)$, contradicting $\dual k\not\in\fn(Q|P')$. 
  
Suppose instead $\dual k\in\domain(\Delta_2)$. Then immediately by Lemma~\ref{lemma-delta-exact-outside-end}
$\dual k\in\fn(P')$, contradicting $\dual k\not\in\fn(Q|P')$.

\Case{\rulename {Com-1}}
We have \[
\frac{ P_1  \transition {\overline{k}!v}
        P_1'\qquad P_2 \transition {k?v} P_2' }
    { P_1 | P_2\transition {\tau}  P_1' | P_2' } 
\]
and
 \[  \frac{
   Θ ⊢_{\standard} P_1  ▹  Δ_1 \qquad Θ ⊢_{\standard} P_2  ▹  Δ_2   }
   {
  Θ ⊢_{\standard} P_1 \PAR P_2  ▹  Δ_1,Δ_2
  }\] 
 By induction we find $ Δ_i\transition  { δ_i} Δ'_i $ s.t.~ $
  Θ  ⊢_{\standard} P_I  ▹  Δ'_i$
  with  $δ_1\correspond\overline {k}  !v$ and $δ_2\correspond k? v$. It
  follows that $δ_1 =k: !$ and $δ_ 2 = k:? $.
   By \rulename{\D-Com1} we thus have $ Δ_1, Δ_2\transition {τ} Δ_1', Δ_2'$.
   Noting that $ τ\correspond  τ$, we have the required type transition. Since $\subject (τ) =  τ$ and so $\dual{\subject (τ)}\not ∈\domain (Δ_1, Δ_2)$,  it follows from Lemma \ref{type-transitions-preserve-balance} that $ Δ_1', Δ_2'$ is balanced when $Δ_1, Δ_2$ is.
   
 \Case{\rulename{\B-Com2}}
  We have \[ \frac{P_1  \transition {\overline{k} ⊕ l} P'_1\qquad  P_2
  \transition {k\& l} 
  P'_2
}{ P_1 | P_2 \transition {\tau:l}  P_1' |  P_2'} \]
  and
 \[  \frac{
   Θ ⊢_{\standard} P_1  ▹  Δ_1 \qquad  Θ ⊢_{\standard} P_2  ▹  Δ_2  }{
  Θ ⊢_{\standard} P_1 \PAR P_2  ▹  Δ_1,Δ_2
  }\] 
  By induction we find $ Δ_i\transition  { δ_i} Δ'_i $ s.t.~ $
  Θ ⊢_{\standard} P_i  ▹  Δ'_i$
  with  $δ_1\correspond\overline {k} ⊕ l$ and $δ_2\correspond k\& l$. It
  follows that $δ_1 =\overline {k}: ⊕ l [L'_1]$, and so we find a transition $
Δ_1, Δ_2\transition {τ: l, L_1' ∪ L_2'} Δ_1', Δ_2'$ by \rulename {\D-Com2}.
 By
  \concr\ we find
  $
  Θ ⊢_{\standard} 
  P'_1 \PAR P'_2  ▹  Δ'_1,Δ'_2
  $.
   Noting that $ τ: l, L_1' ∪ L_2'\correspond  τ:l$ 
we have the required type transition.
Since $\subject (τ: l, L_1' ∪  L_2') = τ$ and so  $\dual{\subject (τ: l, L_1' ∪  L_2')}\not ∈\domain (Δ_1, Δ_2)$,  it follows from Lemma \ref{type-transitions-preserve-balance} that $ Δ_1', Δ_2'$ is balanced when $Δ_1, Δ_2$ is.

\Case{\rulename{\B-Rec}}
 We have
\[ \frac{
                P\{\recp/X\}\transition{λ} Q }{ \recp\transition {λ} Q }
\]
and
\[
  \frac{
     Θ, \AT{X}{Δ} 
    ⊢_{\standard} P  ▹   Δ   }{
     Θ ⊢_{\standard} \recp  ▹   Δ
  }
\]
It then follows by Lemma~\ref{app-lemma-process-variable-substitution} that 
\[
     Θ ⊢_{\standard} P\{\recp/X\}▹ Δ,
\]
and so by induction we find the required  balance-preserving type transition.

\Case{\rulename{\B-Prec0}}
We have 
\[
\frac{ e\Downarrow 0 \qquad Q\transition{ λ}  R }{
\RECN XiePQ  \transition {λ} R } 
\]
and 
\[
\frac{
  Θ, X:Δ 
    ⊢_{\standard} P  ▹  Δ 
    \qquad 
     Θ 
    ⊢_{\standard} R  ▹  Δ 
    }{
     Θ ⊢_{\standard} \recnp  ▹  Δ
  }
  \;, 
\]
and so by the induction hypothesis we find the required balance-preserving type transition.
    
\Case{\rulename{\B-PrecN}}
We have 
\[
\frac{ 
    e\Downarrow n+1\qquad P\{n/i\}\{\RECN XinPQ/X\}\transition {λ} R 
}{
    \recnp\transition {λ} R 
} %
\]
and again
\[
\frac{
  Θ, X: Δ
      ⊢_{\standard} P  ▹  Δ 
    \qquad 
     Θ 
    ⊢_{\standard} Q  ▹  Δ 
     }{
     Θ ⊢_{\standard} \recnp  ▹  Δ
  }
  \;. 
\]
 By Lemma~\ref{app-lemma-variable-substitution} we have
    $ Θ, X:(L ',   Δ) ⊢_{\standard} P\{n/i\}  ▹ Δ$. By Lemma~\ref{app-lemma-process-variable-substitution} we have 
    \[ Θ ⊢_{\standard} P\{n/i\}\{\RECN XinPQ/X\} ▹ Δ, \]
and the requisite balance-preserving type transition follows by the induction hypothesis.

\Case{\rulename{\B-CondT} and \rulename{\B-CondF}}
 We have 
 \[
 \frac{
e \Downarrow \truek\qquad P\transition{λ} P'
}{
\ifp\transition {λ} P'
}
\]
and
\[
  \frac{
        Θ ⊢ P ▹ Δ\qquad
            Θ ⊢ Q ▹ Δ
    }{
      Θ ⊢\ifp ▹ Δ
    }\;,
\]
and the requisite balance-preserving type transition follows from the induction hypothesis.
\end{proof}
\newpage

\fi

\bibliographystyle{entcs} 
\bibliography{db} 

\end{document}


%% file: macros.tex
\newcounter{analphabet}
{\rm%
 \begin{list}%
 {(\arabic{analphabet})}%
 {\usecounter{analphabet}%
  \addtolength{\labelwidth}{5mm}%
  \addtolength{\leftmargin}{-2mm}%
  \setlength{\rightmargin}{0pt}%
  \setlength{\itemsep}{1mm}%
  \setlength{\parsep}{0pt}}}%
 {\end{list}}

\def\keyword#1{\mathsf{#1}}

\def\truek{\keyword{true}}
\def\falsek{\keyword{false}}

\def\ifk{\mathsf{if}}
\def\thenk{\mathsf{then}}
\def\elsek{\mathsf{else}}

\def\PAR{\mid}

\newcommand{\INACT}{\ensuremath{\mathbf 0}}

\newcommand{\SEND}[2]{#1!\langle#2\rangle}

\newcommand{\RECEIVE}[2]{#1?(#2)}
\newcommand{\SEL}[2]{#1!#2.}
\newcommand{\BRANCHS}[3]{#1 ? \{ #2. #3 \}} 
\newcommand{\BRANCH}[2]{#1 ? \{ #2 \}}
\newcommand{\IF}[2]{\ifk\ #1\ \thenk\ #2\ \elsek\ }

\newcommand{\VAR}[2]{#1[#2]}
\newcommand{\RECtriv}[1]{\mathop{\mathrm{rec}}{}#1.}
\let\REC\RECtriv
\newcommand{\FREE}[1]{\ensuremath{\mathsf{free}\,#1}}
\newcommand{\ASSIGN}[2]{\ensuremath{#1:=#2}}
\newcommand{\READ}[2]{\ensuremath{\mathsf{read}\, #1(#2)}}

\newcommand{\RECN}[5]{(\mathop{\mathrm{rec}^{#3}}#1(#2).#4;#5)}
\newcommand{\RECNnp}[5]{\mathop{\mathrm{rec}^{#3}}#1(#2).#4;#5}

\def\sendp{\SEND ke.P}

\def\receivep{\RECEIVE k{x}.P}

\def\selp{\SEL klP}

\def\branchps{\BRANCHS{k}{l_i}{P_i}}    
\def\ifp{\IF ePQ}

\def\varp{\VAR X{\tilde k}}

\def\recp{\REC XP}
\def\recnp{\RECN XiePQ}

\newcommand{\fn}{\operatorname{\mathsf{fn}}}      
\newcommand{\dom}{\operatorname{{\mathsf{dom}}}}

\def\rulename#1{[\text{\sc #1}]\xspace}

\def\sendr{\rulename{\E-Out}}
\def\recvr{\rulename{\E-In}}
\def\branchr{\rulename{\E-Bra}}
\def\selr{\rulename{\E-Sel}}

\def\inactr{\rulename{\E-Inact}}
\def\concr{\rulename{\E-Par}} 
\def\varr{\rulename{\E-Var}}
\def\defr{\rulename{\E-Rec}}
\def\varnr{\rulename{\E-VarP}}
\def\defnr{\rulename{\E-RecP}}

\def\parr{\rulename{Par}}	

\def\ifr{\rulename{\E-Cond}}

\def\branch{\&}

\def\one{\mathtt{end}}

\newcommand{\branchst}[1][T]{\branch\{l_i[L_i].{#1}_i\}_{i\in I}}
\newcommand{\selst}[1][T]{\mathord\oplus\{l_i[L_i].{#1}_i\}_{i\in I}}
\newcommand{\cbranchst}[1][T]{\branch\{l_i{\color{blue}[L_i]}.{#1}_i\}_{i\in I}}
\newcommand{\cselst}[1][T]{\mathord\oplus\{l_i{\color{blue}[L_i]}.{#1}_i\}_{i\in I}}
\def\recvvt{?.T}
\def\sendvt{!.T}

\def\thro wt{\throwT;U}

\def\rect{\mu t.T}

\newcommand{\proves}{\vdash}

\newcommand{\AT}[2]{#1\colon #2}

\def\grmeq{\; ::= \;}
\def\grmor{\; \mid \;}

\newcommand{\Case}[1]{\vspace{.4ex}\noindent\textbf{Case} {#1}\textbf.}

\def\fps@figure{tp}      
\def\fps@table{tp}
\newif\ifny\nytrue
\nyfalse

\newif\ifvv\vvtrue
\vvfalse
\newcommand{\VASCO}[1]
{\ifvv{\color{blue}{#1}}\else{#1}\fi}

\newcommand{\lto}[1]{\xrightarrow{#1}}

\newcommand{\transition}[1]{\lto{#1}}

\newcommand{\infinitetraces}[1]{\mathop{\mathsf{tr}}(#1)}
\newcommand{\selectiontraces}[1]{\mathop{\mathsf{str}}(#1)}
\newcommand{\responsivetraces}[1]{\mathop{\mathsf{tr}_{\!R}}(#1)}
\newcommand{\responsiveselectiontraces}[1]{\mathop{\mathsf{str}_{\!R}}(#1)}

\newcommand{\support}[1]{\mathop{\mathsf{res}}(#1)}
\newcommand{\emphasize}[1]{\emph{#1}}

\newcommand{\N}{\mathbb N}
\renewcommand{\int}{\mathsf{int}}

\newcommand{\dual}[1]{\overline{#1}}
\let\emphasise\emph
\let\domain\dom

\def\ltag#1{\rulename{#1}\quad}
\def\rtag#1{\quad\rulename{#1}}

\newcommand{\responded}{\mathop{\mathsf{res}}}
\newcommand{\requested}{\mathop{\mathsf{req}}}
\newcommand{\selected}{\mathop{\mathsf{sel}}}

\newcommand{\grey}[1]{\fcolorbox{lightgray}{lightgray}{#1}}

%% file: symbols.tex
\newunicodechar{α}{\alpha}
\newunicodechar{β}{\beta}
\newunicodechar{ɣ}{\gamma}
\newunicodechar{Γ}{\Gamma}
\newunicodechar{δ}{\delta}
\newunicodechar{Δ}{\Delta}
\newunicodechar{ε}{\epsilon}
\newunicodechar{ζ}{\zeta}
\newunicodechar{η}{\eta}
\newunicodechar{θ}{\theta}
\newunicodechar{Θ}{\Theta}
\newunicodechar{ι}{\iota}
\newunicodechar{κ}{\kappa}
\newunicodechar{Λ}{\Lambda}
\newunicodechar{λ}{\lambda}
\newunicodechar{μ}{\mu}
\newunicodechar{ν}{\nu}
\newunicodechar{Ξ}{\Xi}
\newunicodechar{ξ}{\xi}
\newunicodechar{Π}{\Pi}
\newunicodechar{π}{\pi}
\newunicodechar{ρ}{\rho}
\newunicodechar{Σ}{\Sigma}
\newunicodechar{σ}{\sigma}
\newunicodechar{τ}{\tau}
\newunicodechar{υ}{\upsilon}
\newunicodechar{Φ}{\Phi}
\newunicodechar{φ}{\phi}
\newunicodechar{χ}{\chi}
\newunicodechar{Ψ}{\Psi}
\newunicodechar{ψ}{\psi}
\newunicodechar{Ω}{\Omega}
\newunicodechar{ω}{\omega}
\newunicodechar{〈}{\langle}
\newunicodechar{⟦}{\llbracket}
\newunicodechar{⟧}{\rrbracket}
\newunicodechar{〉}{\rangle}
\newunicodechar{⊏}{\sqsubset}
\newunicodechar{⊐}{\sqsupset}
\newunicodechar{⊆}{\subseteq}
\newunicodechar{⊇}{\supseteq}
\newunicodechar{≤}{\leq}
\newunicodechar{≥}{\geq}
\newunicodechar{◃}{\triangleleft}
\newunicodechar{▹}{\triangleright}
\newunicodechar{×}{\times}
\newunicodechar{∈}{\in}
\newunicodechar{⊢}{\vdash}
\newunicodechar{→}{\to}
\newunicodechar{⟶}{\longrightarrow}
\newunicodechar{≣}{\equiv}
\newunicodechar{⊕}{\oplus}
\newunicodechar{⊗}{\otimes}
\newunicodechar{⋈}{\bowtie}
\newunicodechar{≝}{\mathbin{\overset{\mathrm{def}}{=}}}
\newunicodechar{∪}{\cup}
\newunicodechar{∩}{\cap}
\newunicodechar{·}{\cdot}
\newunicodechar{∅}{\emptyset}
\newunicodechar{∖}{\setminus}

%% file: FIGSyntax.tex
\def\prefix{M}
\iflong
\par\smallskip
\noindent Meta-variables:
\[
    \begin{array}{ll}
      c & \text{ channel names} \\
      p & \text{ polarities~} +, - \\
      k, h & \text{ polarised channel  names~} (c^p) \\
      x & \text{ data variables}\\
      v & \text{ data values, 
                including natural numbers and $\truek,\falsek$}\\
      e & \text{ data expressions, including data variables and values}\\ 
      l & \text{ selection labels} \\
    X, Y & \text{ process variables} \\
    \end{array}
\]
\else
In general $c$ ranges over channel names; $p$ over polarities $+,-$; $k, h$ over polarised
channel names; $x$ over data variables; $i$ over recursion variables (explained below); $v$ over data values including numbers, strings
and booleans; $ e$ over data expressions; and finally $ X, Y$ over process
variables.
\fi
\iflong
\noindent Process syntax:
\fi
\begin{align*}
   P \grmeq
        & \sendp 
      \grmor  \RECEIVE kx.P  
      \grmor 
			  \SEL klP  
      \grmor  {\BRANCHS k{l_i}{P_i}}_{i\in I} 
   \grmor \INACT \grmor P\! \PAR\! Q  
      \\
      \grmor &         \recp
                 \grmor  \recnp 
                 \grmor  \varp
                 \grmor \ifp
  \end{align*}

%

%% file: FIGTypes.tex
%
\iflong\relax\else
Let $l$ range over labels and $L$ sets of labels. 
\fi
\begin{gather*}
      \begin{array}{ll}
    \iflong
      \mathcal L & \;\text{a countably infinite set of labels} \\
      l & \;\text{ranges over~$\mathcal L$} \\
      L &  \;\text{ranges over~$\mathcal P(\mathcal L)$}
  \\\strut
  \\
  \fi
  S,T &\grmeq 
     \cbranchst \grmor 
	\cselst 
\grmor \sendvt \grmor \recvvt	
	\grmor 
 \rect \grmor t 
	\grmor  \one 
  \end{array}
\end{gather*}
%

%% file: FIGtype-system-std.tex
\def\C{\ref{fig:std-type-system}}
\begin{figure}[htb]
  \centering  
  \begin{gather*}
    \ltag{\C-Out}
    \frac{
        Θ  ⊢_{\standard} P  ▹  Δ, k: T \qquad
    }{
        Θ ⊢_{\standard} \sendp  ▹  Δ, k: !.T
    }
    \qquad
    \frac{
        Θ  ⊢_{\standard} P  ▹  Δ, k: T
    }{
        Θ ⊢_{\standard} \receivep  ▹  Δ, k: ?.T
    }
    \rtag{\C-In}
    \\[6mm]	
    \ltag{\C-Bra}
    \frac{
       \forall i \in I: \quad   Θ  ⊢_{\standard} P_i  ▹  Δ , k: T_i	        
    }{
         Θ ⊢_{\standard}\ {\BRANCHS k{l_i}{P_i}}_{i\in I}   ▹  Δ, k:\branchst 
    }
\iflong \\[6mm] \ltag{\C-sel} \else \qquad \fi
    \frac{
        Θ  ⊢_{\standard} P  ▹  Δ, k: T_j
    }{
        Θ ⊢_{\standard} \SEL{k}{l_j}{P}  ▹  Δ, k:\selst 
    }
    \quad(j\in I)
\iflong \else \rtag{\C-Sel} \fi
    \\[6mm]
    \ltag{\C-Par}
    \qquad
  \frac{
     Θ ⊢_{\standard} P_1  ▹  Δ_1 \qquad   Θ ⊢_{\standard} P_2  ▹  Δ_2
  }{
   Θ ⊢_{\standard} P_1 \PAR P_2  ▹  Δ_1,Δ_2
  }
  \qquad
 \frac{
      \VASCO{Δ\text{ completed}}
    }{
        Θ ⊢_{\standard} \INACT  ▹  Δ
    }      
        \qquad\rtag{\C-Inact}
  \\[6mm]
  \frac{
      Θ , \AT{X}{Δ}
    ⊢_{\standard} P  ▹  Δ 
    \qquad 
      Θ 
    ⊢_{\standard} Q  ▹  Δ 
  }{
      Θ ⊢_{\standard} \recnp  ▹  Δ
  } 
  \rtag{\C-RecP}
   \\[6mm] 
   \ltag{\C-Rec}
  \frac{
      Θ , \AT{X}{Δ}
    ⊢_{\standard} P  ▹  Δ 
  }{
      Θ ⊢_{\standard} \recp  ▹  Δ
  }  
  \qquad
    \frac{
         \domain (Δ) =\tilde k
      }{
         Θ ,\AT{X}{Δ}
       ⊢_{\standard} \varp  ▹  Δ
     }
     \rtag{\C-Var}
     \\[6mm]
 \ltag{\C-Cond}    
    \frac{
     Θ ⊢_{\standard}  P ▹ Δ\qquad
          Θ ⊢_{\standard}  Q ▹ Δ
    }{
    Θ ⊢_{\standard}\ifp ▹ Δ
    }
   \end{gather*}
  \caption{Standard Session Typing System}
  \label{fig:std-type-system}
\end{figure}

%% file: FIGtype-system.tex
\def\E{\ref{fig:type-system}}
\begin{figure}[htbp]
    \shortcut
      \def\ltag#1{#1\quad}
  \centering  
  \begin{gather*}
       \ltag{\sendr}
    \frac{
       Γ ; L  ⊢ P  ▹  Δ, k: T
    }{
       Γ; L  ⊢ \sendp  ▹  Δ, k: !.T
    }
    \qquad\qquad
    \ltag{\recvr}
    \frac{
       Γ ; L ⊢ P  ▹  Δ, k: T
    }{
       Γ; L  ⊢ \receivep  ▹  Δ, k:?.T
    }
    \\[4mm]	
    \ltag{\branchr}
    \frac{
       \forall i \in I: \quad  Γ + l_i; (L  ∖ l_i) ∪ L_i  ⊢ P_i  ▹  Δ , k: T_i	        
    }{
         Γ; L  ⊢\ {\BRANCHS k{l_i}{P_i}}_{i\in I}  ▹  Δ, k:\branchst 
    }
    \\[4mm]
    \ltag{\selr}
    \frac{
       Γ + l_j; (L  ∖ l_j) ∪ L_j ⊢ P  ▹  Δ, k: T_j
    }{
       Γ; L  ⊢ \SEL{k}{l_j}{P}  ▹  Δ, k:\selst 
    }
    \quad(j\in I)
    \\[4mm]
  \ltag{\concr}
  \frac{
   Γ; L _1 ⊢ P_1  ▹  Δ_1 \qquad  Γ; L _2 ⊢ P_2  ▹  Δ_2 
  }{
   Γ; L _1 ∪ L _2 ⊢ P_1 \PAR P_2  ▹  Δ_1,Δ_2
%
  }
  \qquad
   \ltag{\inactr}
    \frac{
      \VASCO{Δ\text{ completed}}
    }{
       Γ; ∅ ⊢ \INACT  ▹  Δ
    }      
  \\[4mm]
    \ltag{\varr}
    \frac{
 L ⊆I  ⊆  A    \qquad  \domain (Δ) =\tilde k }{
        Γ,X:(A, I,   Δ);  L 
       ⊢ \varp  ▹  Δ
     }
\qquad\qquad
  \ltag{\varnr}
    \frac{
L   ⊆  L '  \qquad  \domain (Δ) =\tilde k     }{
        Γ,X:(L ',   Δ);  L 
       ⊢ \varp  ▹  Δ
     }
     \\[4mm]
  \ltag{\defnr} 
    \frac{
     Γ, X:(L ',   Δ);  L'
    ⊢ P  ▹  Δ 
    \qquad 
     Γ;  L'  
    ⊢ Q  ▹  Δ 
    \qquad
    L ⊆ L'
  }{
     Γ; L  ⊢ \recnp  ▹  Δ
  } 
   \\[4mm] 
  \ltag{\defr} 
  \frac{
     Γ, \AT{X}{(∅, I,   Δ)};  I
    ⊢ P  ▹   Δ  \iflong\qquad\else\quad\fi  L  ⊆  I
  }{
     Γ; L  ⊢ \recp  ▹   Δ
  }  
\iflong \\[4mm] \else \quad \fi
   \ltag{\ifr}   
    \frac{
      Γ; L ⊢ P ▹ Δ\iflong\qquad\else\quad\fi
           Γ; L ⊢ Q ▹ Δ
    }{
     Γ; L ⊢\ifp ▹ Δ
    }
    \end{gather*}
    \shortcut
      \caption{Typing System}
      \shortcut 
  \label{fig:type-system}
  \end{figure}